\documentclass[conference]{IEEEtran}
\IEEEoverridecommandlockouts

\usepackage{cite}
\usepackage{amsmath,amssymb,amsfonts,amsthm}
\usepackage{algorithmic}
\usepackage{graphicx}
\usepackage{textcomp}
\usepackage{xcolor}
\usepackage{multirow} 
\usepackage{subfigure}
\usepackage{hyperref}
\usepackage{bm}
\usepackage{verbatim}
\usepackage{booktabs}
\usepackage[linesnumbered,boxed,ruled,commentsnumbered]{algorithm2e}
\usepackage{balance}
\usepackage{endnotes}
\usepackage{bbding}
\usepackage{enumitem}

\newcommand{\corrmark}{\textsuperscript{\Envelope}}

\newtheorem{myDef}{Definition}

\newtheorem{myPro}{Proposition}

\usepackage{amsthm}

\theoremstyle{plain}
\newtheorem{myLem}{Lemma}

\def\BibTeX{{\rm B\kern-.05em{\sc i\kern-.025em b}\kern-.08em
    T\kern-.1667em\lower.7ex\hbox{E}\kern-.125emX}}

\begin{document}

\title{Risk-aware Road-network Trajectory Synthesis under Local Differential Privacy}

\author{Hongzhen Xiang$^{1}$, Chenxu Wang$^{1,3}\corrmark$, Yumeng Song$^{2}$, Tao Qin$^{3}$, Yushuai Li$^{2}$, Pinghui Wang$^{3}$,  Wei Wang$^{1}$, 

Tianyi Li$^{2}$\\
 {\large$^{1 }$}School of Software Engineering, Xi’an Jiaotong University, Xi’an, China\\
 {\large$^{2}$} Department of Computer Science, Aalborg University, Aalborg, Denmark   \\
 {\large$^{3}$} MoE Key Lab of Intelligent Networks and Network Security, Xi’an Jiaotong University, Xi’an, China   \\
{\large$^{1}$} hzxiang@stu.xjtu.edu.cn, 
{\large$^{1,3}$}cxwang@mail.xjtu.edu.cn,  
{\large$^{2}$} \{yumengs, yusli, tianyi\}@cs.aau.dk,\\
{\large$^{3}$} \{qin.tao, phwang\}@mail.xjtu.edu.cn,
{\large$^{1}$} wei.wang@mail.xjtu.edu.cn
}

\maketitle

\begingroup
\renewcommand{\thefootnote}{\Envelope}
\footnotetext[1]{Corresponding author.}
\endgroup

\begin{abstract}
Road-network trajectory synthesis under local differential privacy (LDP) enables mobility data analysis without uploading raw trajectories to an untrusted server. Existing locally private trajectory synthesis methods mainly model whole trajectories, point sequences, or grid-level mobility patterns. However, they cannot handle non-uniform privacy risks caused by local road-network and trajectory contexts, such as endpoints, long-stay-associated regions, and low-out-degree branches. They also do not directly protect conditional next-hop choices on road networks, where a user's selected successor may reveal sensitive local movement intentions. 
To address these challenges, we propose \textsc{TrajRACE}, a locally private synthesis framework for road-network trajectories. \textsc{TrajRACE} uses start, segment-count, and transition reports to represent trajectories. For each transition, it keeps the current road segment as a public aggregation context and perturbs the true successor among legal successors on the client side. It then assigns risk-aware protection strengths to transitions according to local road-network and trajectory contexts. Finally, the server recovers population-level start, segment-count, and conditional transition statistics from heterogeneous privatized reports, and generates synthetic trajectories that follow the public road-network topology.
Experiments on three road-network trajectory datasets demonstrate the effectiveness of \textsc{TrajRACE} in improving synthetic-data utility across distributional, query-based, and pattern-level metrics, while providing stronger protection for high-risk local transitions. 
\end{abstract}


\section{Introduction}\label{sec:introduction}
With the widespread adoption of smartphones, in-vehicle navigation devices, and location-aware terminals, trajectory data have become an important foundation for traffic monitoring~\cite{fan2013traffic,li2021trace,li2022evolutionary}, route planning, location-based services, and urban computing~\cite{jiang2021lbsurvey,zhang2024facility,zheng2014urban,hu2024similarity,li2020uncertain}. However, trajectory data contain continuous and highly personalized mobility behaviors. Such data may reveal users' activity regions, travel habits, and social relations~\cite{montjoye2013unique,chang2018privacy}. Under an untrusted server setting, directly uploading raw trajectories can cause serious privacy risks. This risk may remain even after anonymization, aggregation, or statistical release~\cite{pyrgelis2018membership,wang2018untrusted,ahuja2023spatiotemporal,jin2022survey}. 

Local differential privacy (LDP) has become a widely adopted solution for trajectory privacy protection. Given a set of trajectories, LDP~\cite{dwork2006dp,cormode2018ldp,xiong2020ldpsurvey} protects trajectory data on the client side before uploading. A raw trajectory is decomposed into local data units, such as trajectory points, directions, fragments, starts, transitions, or lengths. Each unit is perturbed and uploaded as a privatized report. Thus, the server cannot directly access raw trajectories and only observes privatized reports. The server aggregates these reports to estimate population-level mobility statistics and generate synthetic trajectories for downstream analysis. Existing studies~\cite{cunningham2021,zhang2023,hu2024} improve LDP-based trajectory synthesis by extracting useful mobility information from perturbed data. Cunningham \textit{et al}.~\cite{cunningham2021} incorporate public knowledge to improve the realism and utility of synthesized trajectories. Zhang \textit{et al}.~\cite{zhang2023} leverage adjacent direction information and adaptive spatial restrictions to improve trajectory collection under pure $\epsilon$-LDP. Hu \textit{et al}.~\cite{hu2024} propose RetraSyn, which supports real-time trajectory synthesis from privately extracted mobility patterns. 

\begin{figure}[t]
  \centering
  \includegraphics[width=\linewidth]{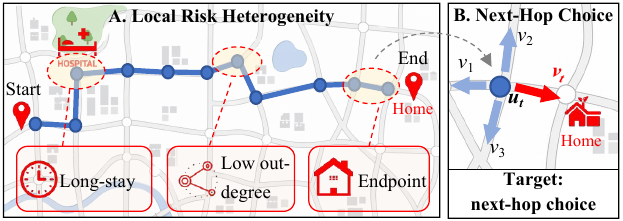}
  \vspace{-6mm}
  \caption{Example of non-uniform risks and next-hop protection.}
  \label{fig:motivation}
  \vspace{-7mm}
\end{figure}

However, several key challenges have yet to be addressed.

\noindent\textit{Challenge I: How to handle non-uniform privacy risks within a trajectory?}
{ Existing LDP-based trajectory synthesis methods~\cite{cunningham2021,zhang2023,hu2024} usually apply similar protection strength to different trajectory points or local transitions.} However, privacy risks are not evenly distributed within a real road-network trajectory. As shown in Fig.~\ref{fig:motivation}(A), a trajectory may pass through an ordinary main road, move near a hospital, enter a low-out-degree branch, and finally stop near a residential area. {A transition near a hospital-related area} may reveal a health-related visit, the next hop after the low-out-degree branch may narrow down the destination, and the residential endpoint may expose the user's home region. In contrast, ordinary main-road { transitions} usually only indicate common commuting flows. If all {transitions} within a trajectory receive the same protection strength, high-risk {transitions} may be under-protected. For example, the server may still recover a clear group-level transition pattern to the hospital area from many perturbed next-{hop} choices. This does not reveal one user's exact trajectory, but it may expose sensitive mobility patterns around the hospital. Meanwhile, low-risk main-road { transitions} may be over-perturbed, which scatters common commuting flows to incorrect successors. As a result, treating all {transitions} in a trajectory equally may weaken the protection of sensitive movements and reduce the utility of synthesized trajectories.

\noindent\textit{Challenge II: How to protect local next-hop decisions?}
A user's next-hop choice at a branching road segment may reveal sensitive information. As shown in Fig.~\ref{fig:motivation}(B), at time $t$,  {given the current road segment $u_t$, the} user has four {possible successors,} i.e., $v_1$, $v_2$, $v_3$, and $v_t$. These  {successors} may have different privacy risks. For example, $v_1$, $v_2$, and $v_3$ may correspond to common public routes, while $v_t$ may point to the user's home area. If the recovered transition distribution assigns a high probability to $v_t$, {the next-hop pattern may} become more predictable. This may reveal sensitive local movement tendencies without revealing complete trajectories. In trajectory synthesis under LDP, this predictability may remain after aggregation. Although each user's data is perturbed, a dominant next-hop pattern may still appear after many perturbed samples are aggregated. The synthesized trajectories may then reproduce this sensitive direction too often. Thus, protected synthetic data may still reveal where users are likely to go after a branching road segment, especially when the next hop leads to sensitive places such as homes, hospitals, or private activity regions.

To address the above challenges, we propose \textsc{TrajRACE}, a locally private trajectory synthesis framework for road-network trajectories under LDP. \textsc{TrajRACE} protects users' next-hop choices under road-network context and adapts protection to non-uniform privacy risks within trajectories. It then generates privacy-preserving synthetic trajectories from locally perturbed {trajectory reports}.
To address the first challenge, \textsc{TrajRACE} handles non-uniform privacy risks within trajectories by assigning different protection strengths to different transitions. It evaluates transition risks using road-network topology and local trajectory context. High-risk transitions, such as movements { near sensitive regions, low-out-degree branches, or endpoints,}  receive stronger protection. Low-risk transitions, such as ordinary pass-through movements on main roads, avoid unnecessary perturbation. This helps protect sensitive local movements while preserving common mobility patterns.
To address the second challenge, \textsc{TrajRACE} protects the true next-hop choice under a given road-segment context. For each transition, it keeps the current road segment as {a public} aggregation context, but randomizes the chosen successor among legal successors on the client side. This allows the server to estimate transition distributions for each current road segment, while preventing it from directly observing the true next hop. This {mitigates the direct disclosure} of sensitive next-hop choices and supports trajectory synthesis from protected {conditional} transition statistics.

We summarize our contributions as follows:
\begin{itemize}[leftmargin=3mm]
\item We propose \textsc{TrajRACE}, a locally private framework for road-network trajectory synthesis under the LDP setting. It generates road-network-consistent synthetic trajectories from privatized {reports}, while protecting next-hop choices and handling non-uniform privacy risks within trajectories.

\item We design a risk-aware event-level protection mechanism {for transition reports} to handle non-uniform privacy risks. It uses road-network topology and local trajectory context to assign stronger protection to high-risk transitions.

\item We develop a bucket-corrected heterogeneous recovery and synthetic release method. It recovers population-level start, segment-count, and conditional transition statistics from privatized reports and generates road-network-consistent synthetic trajectories.

\item We conduct experiments on three road-network trajectory datasets. The results show that \textsc{TrajRACE} improves synthetic-data utility, reduces exposure of high-risk successors, and validates the effectiveness of the key designs.
\end{itemize}

This paper is organized as follows.
Section~\ref{sec:related_work} reviews related work.
Section~\ref{sec:preliminary} presents the preliminaries and problem definition.
Section~\ref{sec:methodology} presents the \textsc{TrajRACE} framework and analyzes its properties and complexity.
Section~\ref{sec:experiments} reports the experimental results.
Section~\ref{sec:conclusion} concludes the paper.
\section{Related Work}\label{sec:related_work}

\subsection{Privacy-preserving Trajectory Synthesis}
Privacy-preserving trajectory synthesis can be studied under centralized differential privacy (DP) and LDP settings. 

Centralized DP assumes that a trusted server can access raw data and then releases protected statistics or synthetic data.  Existing studies generate synthetic data by learning protected population-level statistics from raw data under this setting~\cite{he2015dpt,bindschaedler2016,gursoy2018adatrace,cai2021dpmrf,ge2021kamino,zhang2014privbayes,mckenna2022aim}. These studies include DP trajectory synthesis, graphical-model-based synthesis, Bayesian-network-based synthesis, and general private database synthesis. They improve synthetic data quality by modeling correlations among attributes, locations, or mobility patterns. PrivBayes~\cite{zhang2014privbayes} uses Bayesian networks for private tabular data synthesis, AdaTrace~\cite{gursoy2018adatrace} generates synthetic trajectories from DP-protected mobility patterns, and AIM~\cite{mckenna2022aim} adaptively selects useful measurements for private data synthesis. These studies show that synthetic data can preserve useful statistics under DP. However, these methods require a trusted data collector, while this paper considers an untrusted server setting where users protect data before uploading.

In contrast, LDP does not require a trusted server. Users perturb their trajectory data locally, and the server only aggregates the privatized {reports} for analysis and synthesis. Under the LDP setting, studies usually focus on two related problems. The first is how users perturb trajectory data locally and how the server collects useful privatized data~\cite{du2023ldptrace,yang2022collect,ren2022ldpids,wang2021continuous,li2023streaming,erlingsson2014rappor,wang2017ldpfrequency}. The key is to design LDP mechanisms that protect users' locations, transitions, and trajectory lengths while preserving enough information for analysis. Cunningham \textit{et al}.~\cite{cunningham2021} incorporate public knowledge to improve the realism and utility of LDP trajectory release. Zhang \textit{et al}.~\cite{zhang2023} study trajectory data collection under strict pure $\epsilon$-LDP while reducing the reliance on external knowledge. The second problem is how to recover statistics from privatized data and generate synthetic data. Hu \textit{et al}.~\cite{hu2024} propose RetraSyn for real-time trajectory synthesis from private trajectory streams. R.Du \textit{et al}.~\cite{du2025privacyfree} study statistical recovery from LDP-perturbed data collected by multiple services, aiming to improve mean and distribution estimation without consuming additional privacy budget.

These studies improve privacy-preserving trajectory synthesis through public knowledge, local perturbation, online synthesis, and statistical recovery. However, they do not explicitly protect next-hop choices under road-network context, nor do they adapt protection to non-uniform privacy risks along a trajectory. Sensitive movements may remain exposed, while ordinary movements may be over-perturbed. {\textsc{TrajRACE} centers on conditional next-hop choices and supports risk-aware transition protection, bucket-corrected heterogeneous recovery, and road-network-consistent synthetic release.}

\subsection{Non-uniform Trajectory Privacy Protection}
Non-uniform privacy protection aims to avoid assigning the same protection strength to all data units~\cite{du2021ahead,zhang2018calm,cao2017temporal,xiao2015temporal,schaler2023wevent,jorgensen2015personalized,he2014blowfish}. The main idea is to give stronger protection to sensitive units and weaker protection to less sensitive units, thereby improving the privacy-utility trade-off. Existing studies define non-uniformity from different perspectives. Wei \textit{et al}.~\cite{wei2024aaa} focus on adaptive perturbation for users. Their AAA mechanism adjusts the local perturbation strength under different conditions to improve overall utility. L.Du \textit{et al}.~\cite{du2025personalized} focus on temporal non-uniformity. They study personalized \(w\)-event privacy, where different events or time windows can have different privacy requirements. R.Du \textit{et al}.~\cite{du2025privacyfree} focus on non-uniformity in statistical recovery. The LDP-perturbed data collected by different services may follow different perturbation mechanisms and privacy budgets. They distinguish such differences during recovery and reuse these perturbed data to improve mean and distribution estimation without extra privacy cost.

However, these studies mainly model non-uniformity from users, events, time windows, services, or perturbation mechanisms. They do not address non-uniform privacy risks within road-network trajectories. In a trajectory, local movements may carry different privacy risks. Existing methods do not explicitly model road-network-aware intra-trajectory transition risks. \textsc{TrajRACE} addresses this issue by measuring transition risks from road-network topology and trajectory context, and by giving stronger protection to high-risk transitions.

\section{PRELIMINARIES}\label{sec:preliminary}

 \subsection{Road-segment Trajectory}

\begin{myDef}
A \textbf{raw trajectory} is a time-ordered sequence of GPS points, denoted as $T^r=\langle p_1,p_2,\ldots,p_m\rangle$, where each point $p_i=(x_i,y_i,t_i)$ records the longitude $x_i$, the latitude $y_i$, and the timestamp $t_i$ of a moving object. The timestamps satisfy $t_i<t_{i+1}$ for $1\leq i<m$.
\end{myDef}

\begin{myDef}
A \textbf{road network} is modeled as a directed graph $G=(V,\mathcal E)$, where $V$ is a set of road junctions and $\mathcal E\subseteq V\times V$ is a set of directed road segments. Each road segment $e=(v_i,v_j)\in\mathcal E$ starts at junction $v_i$ and ends at junction $v_j$.
\end{myDef}

Map-matching~\cite{newson2009hidden} aligns a raw trajectory with the road network and constrains the movement of the object to road segments. It maps each GPS point to a network-constrained location. The trajectory is represented on the road network.

\begin{myDef}
Given a road segment $e=(v_i,v_j)\in\mathcal E$, a \textbf{mapped location} is denoted as $l=(e,\delta,t)$, where $e$ is the road segment containing $l$, $\delta$ is the network distance from the start junction $v_i$ to $l$ along $e$, $0\leq\delta\leq |e|$, and $t$ is the timestamp inherited from the original GPS point. Here, $|e|$ is the length of road segment $e$.
\end{myDef}

\begin{myDef}
Given a raw trajectory $T^r=\langle p_1,p_2,\ldots,p_m\rangle$, its \textbf{map-matched trajectory} is a time-ordered sequence of mapped locations, denoted $T^m=\langle l_1,l_2,\ldots,l_m\rangle$, where $l_k$ is the mapped location of the GPS point $p_k$ on the road network.
\end{myDef}

\begin{myDef}
Given a road network $G=(V,\mathcal E)$ and a road segment $e=(v_i,v_j)\in\mathcal E$, the \textbf{legal successor set} of $e$ is denoted as $\mathcal N(e)$. It contains all road segments that can be directly reached after traveling through $e$, i.e., $\mathcal N(e)=\{(v_j,v_k)\in\mathcal E \mid v_k\in V\}$.
\end{myDef}

\begin{myDef}
Given a map-matched trajectory $T^m=\langle l_1,l_2,\ldots,l_m\rangle$, its \textbf{road-segment trajectory} is denoted as $T=\langle l_s,g_1,g_2,\ldots,g_n,l_d\rangle$, where $l_s=l_1$, $l_d=l_m$, and each $g=(e,\tau)$ is a time-stamped road segment. For each $g_i$, $g_i.e$ denotes its road segment and $g_i.\tau$ denotes its timestamp. The road-segment sequence $\langle g_1.e,g_2.e,\ldots,g_n.e\rangle$ is obtained by merging consecutive mapped locations in $T^m$ that lie on the same road segment, while preserving their travel order. For $1\leq i<n$, $g_{i+1}.e\in\mathcal N(g_i.e)$. The timestamp $g_i.\tau$ is the average timestamp of all mapped locations in the corresponding consecutive group on road segment $g_i.e$.
\end{myDef}

\subsection{Problem Statement}

\begin{myDef}
Given a road-segment trajectory
$T=\langle l_s,g_1,g_2,\ldots,g_n,l_d\rangle$ and a public segment bound $L_m$, the \textbf{segment-bounded subtrajectory (SBS) set}  of $T$ is obtained by partitioning $T$ into subtrajectories, each covering at most $L_m$ road segments. It is denoted as
$\mathcal S(T)=\{S_1,S_2,\ldots,S_q\}$, where
$S_j=\langle l_{s,j},g_{a_j},g_{a_j+1},\ldots,g_{r_j},l_{d,j}\rangle$
$(1\leq j\leq q)$. Here, $l_{s,j}$ and $l_{d,j}$ are the start and destination mapped locations of $S_j$, and $g_{a_j},g_{a_j+1},\ldots,g_{r_j}$ are the time-stamped road segments covered by $S_j$ in their travel order. For $1\leq j\leq q$, the number of covered time-stamped road segments satisfies $r_j-a_j+1\leq L_m$. If $l_{s,j}$ or $l_{d,j}$ does not correspond to an original mapped location in $T^m$, its timestamp is set to the timestamp of the time-stamped road segment on which it lies.
\end{myDef}

\begin{myDef}
Given a segment-bounded subtrajectory
$S_j=\langle l_{s,j},g_{a_j},g_{a_j+1},\ldots,g_{r_j},l_{d,j}\rangle$,
its \textbf{local reporting units} are denoted as
$\mathcal I(S_j)=\{\textit{st}_j,\textit{cnt}_j,\mathcal X_j\}$, where
$\textit{st}_j=g_{a_j}.e$ is the start road segment containing $l_{s,j}$,
$\textit{cnt}_j=r_j-a_j+1$ is the number of time-stamped road segments in $S_j$, and
$\mathcal X_j=\langle (g_i.e,g_{i+1}.e)\rangle_{i=a_j}^{r_j-1}$ is the ordered transition-unit sequence. For each transition unit $(g_i.e,g_{i+1}.e)$, $g_i.e$ is treated as a disclosed public conditioning context, while $g_{i+1}.e\in\mathcal N(g_i.e)$ is the protected true successor.
\end{myDef}

\begin{myDef}
Given a domain $\mathcal X$, an algorithm $\mathcal M$ satisfies \textbf{$\epsilon$-local differential privacy} ($\epsilon$-LDP), where $\epsilon\geq0$ is the privacy budget, if and only if for any $o,o'\in\mathcal X$ and output $z$:
\begin{equation}
    \Pr(\mathcal M(o)=z)\le e^\epsilon \Pr(\mathcal M(o')=z),
\end{equation}
\end{myDef}

The privacy budget $\epsilon$ bounds how much the output of a perturbation mechanism can reveal about the user's true input. A smaller $\epsilon$ makes different input values harder to distinguish from the perturbed output, and provides privacy protection.

\begin{myDef}
Given a trajectory dataset $\mathcal D$ on a road network $G=(V,\mathcal E)$ and a public segment bound $L_m$, each trajectory $T\in\mathcal D$ is first partitioned into its SBS set $\mathcal S(T)$, yielding $\mathcal D^b=\bigcup_{T\in\mathcal D}\mathcal S(T)$. The \textbf{population-level mobility statistics} of $\mathcal D$ are defined as $\Phi(\mathcal D)=\{p_s,\,p_c,\,\{p_e\}_{e\in\mathcal U}\}$, where $p_s$ is the distribution of start units $\textit{st}_j$, and $p_c$ is the distribution of segment-count units $\textit{cnt}_j$ for all $S_j\in\mathcal D^b$. The set $\mathcal U\subseteq\mathcal E$ contains all road segments that appear as current segments in transition units. For each $e\in\mathcal U$, $p_e$ denotes the successor distribution of $e$, defined as $p_e(e')=\Pr(e'\mid e)$ for $e'\in\mathcal N(e)$.
\end{myDef}

\begin{myDef}
Given a public road network $G=(V,\mathcal E)$, a public segment bound $L_m$, an SBS reporting-accounting cap $B$, and a trajectory dataset $\mathcal D$ collected on the clients, let $\mathcal D^b=\bigcup_{T\in\mathcal D}\mathcal S(T)$ be the corresponding SBS dataset. The \textbf{conditional next-hop private trajectory synthesis (CNPTS)} problem is to design a client-side perturbation mechanism $\mathcal M$ and a server-side synthesis method $\mathcal F$ such that:
\begin{enumerate}[label=(\roman*), leftmargin=11pt]
\item For each $S_j\in\mathcal D^b$, let $R_j=|\mathcal X_j|$. The client applies $\mathcal M$ to its local reporting units $\mathcal I(S_j)=\{\textit{st}_j,\textit{cnt}_j,\mathcal X_j\}$ and generates privatized reports $\mathcal M(S_j)=\{\hat{st}_j,\hat{cnt}_j,\{\rho_{j,t}\}_{t=1}^{R_j}\}$. The start report $\hat{st}_j$ satisfies $\epsilon_s$-LDP for $\textit{st}_j$, and the segment-count report $\hat{cnt}_j$ satisfies $\epsilon_c$-LDP for $\textit{cnt}_j$. For each transition unit $(g_i.e,g_{i+1}.e)\in\mathcal X_j$, $g_i.e$ is a disclosed public conditioning context and the true successor $g_{i+1}.e$ is protected by a conditional transition-report mechanism.

\item For each $S_j$, if $k_{j,t}$ is the execution bucket of the $t$-th transition report, its reporting-accounting cost is $\epsilon_x^{(k_{j,t})}=\epsilon_e^{(k_{j,t})}+\epsilon_b$, and the realized SBS reporting-accounting cost satisfies $\epsilon_s+\epsilon_c+\sum_{t=1}^{R_j}\epsilon_x^{(k_{j,t})}\le B$.

\item The server only receives privatized reports and generates a synthetic dataset $\mathcal D^\ast=\mathcal F\!\left(\{\mathcal M(S_j)\mid S_j\in\mathcal D^b\}\right)$. The synthetic dataset $\mathcal D^\ast$ consists of road-segment trajectories on $G$, and
$\Phi(\mathcal D^\ast)$ is as close as possible to $\Phi(\mathcal D)$.
\end{enumerate}
\end{myDef}

\section{METHODOLOGY}\label{sec:methodology}

\begin{figure*}[htbp]
  \centering
  \includegraphics[width=0.9\linewidth,height=0.285\textheight]{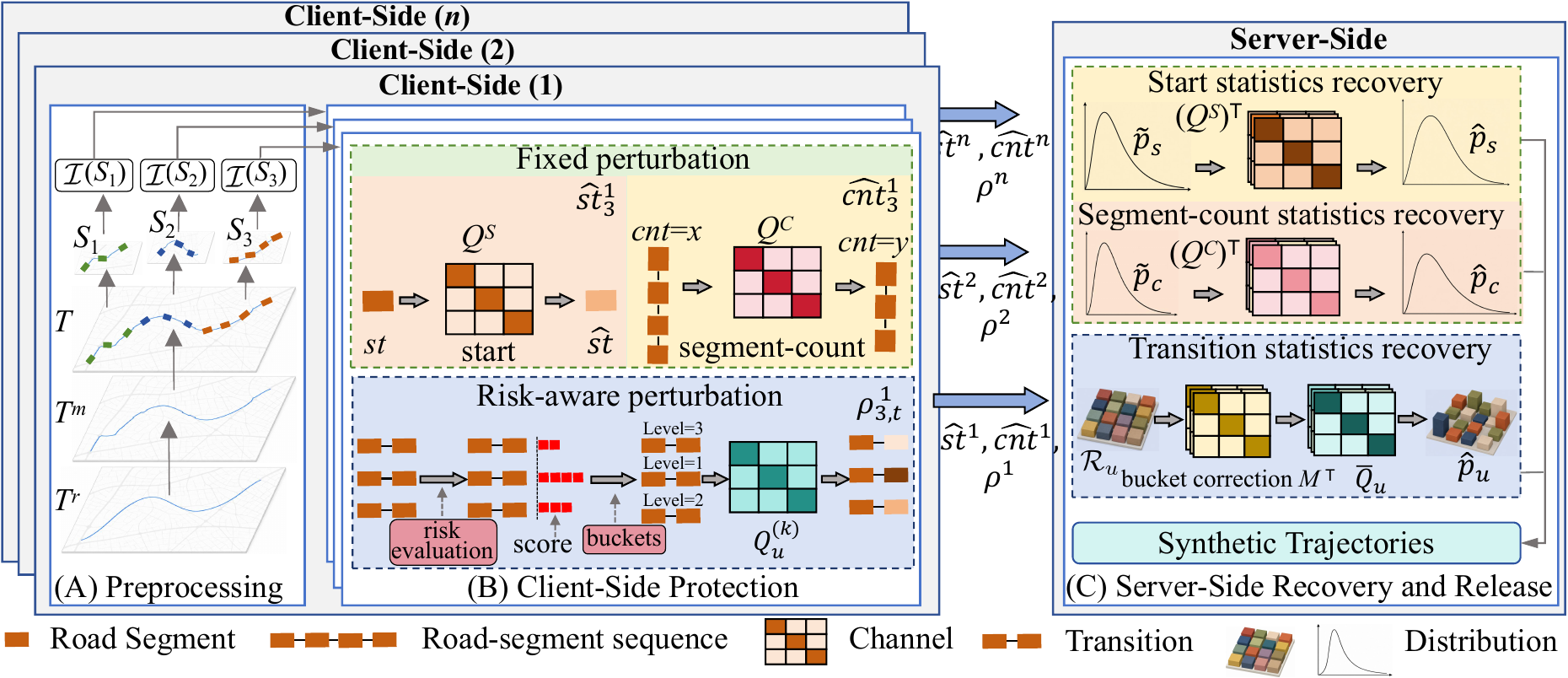}
  \vspace{-10pt}
  \caption{{Framework of \textsc{TrajRACE}: (A) trajectory preprocessing, (B) client-side protection, and (C) server-side recovery and release.}}
  \vspace{-20pt}
  \label{fig:framework}
\end{figure*}

\subsection{Method Overview}

As shown in Fig.~\ref{fig:framework}, \textsc{TrajRACE} consists of three modules: trajectory preprocessing and protection on the client side, followed by recovery and release on the server side.

\noindent\textbf{Trajectory Preprocessing.}
Given a trajectory $T^r$, \textsc{TrajRACE} map-matches it~\cite{newson2009hidden} to obtain $T^m$, converts it into a road-segment trajectory $T$, and partitions $T$ into $\mathcal S(T)$ under the public segment bound $L_m$. For each $S_j\in\mathcal S(T)$, it extracts the local reporting units $\mathcal I(S_j)=\{\textit{st}_j,\textit{cnt}_j,\mathcal X_j\}$. We detail this module in Section~\ref{sec:eventization}.

\noindent\textbf{Client-Side Protection.}
\textsc{TrajRACE} applies GRR to $\textit{st}_j$ and $\textit{cnt}_j$, and risk-aware conditional perturbation to each $(g_i.e,g_{i+1}.e)\in\mathcal X_j$, where $g_i.e$ is the public current segment and $g_{i+1}.e$ is the protected true successor. It outputs privatized start, segment-count, and transition reports. We detail this module in Section~\ref{sec:client_protection}.

\noindent\textbf{Server-Side Recovery and Release.}
The server recovers population-level start, segment-count, and conditional transition statistics from privatized reports, and generates a road-network-consistent synthetic dataset $\mathcal D^\ast$. We detail this module in Section~\ref{sec:server_recovery}.

\subsection{Trajectory Preprocessing}\label{sec:eventization}
For a raw trajectory $T^r$, \textsc{TrajRACE} does not use all GPS points for synthesis. Instead, it extracts a small set of local reporting units that are sufficient for recovering population-level mobility statistics. 
As shown in Fig.~\ref{fig:framework}(A), \textsc{TrajRACE} first map-matches $T^r$ to the public road network and obtains a map-matched trajectory $T^m$. It then converts $T^m$ into a road-segment trajectory $T$ and partitions $T$ into an SBS set $\mathcal S(T)$ under the public segment bound $L_m$. For each SBS $S_j\in\mathcal S(T)$, \textsc{TrajRACE} extracts its local reporting units $\mathcal I(S_j)=\{\textit{st}_j,\textit{cnt}_j,\mathcal X_j\}$, including the start road segment, the number of road segments in the SBS, and local transitions. These units are privacy-sensitive because they may reveal where a user starts, how many road segments the user travels within an SBS, and which next-hop choices the user makes. They are also useful for statistical recovery because they directly correspond to the statistics needed for synthetic trajectory generation. \textsc{TrajRACE} perturbs these units on the client side and uploads only the resulting privatized reports.

\subsection{Client-Side Protection}\label{sec:client_protection}
As shown in Fig.~\ref{fig:framework}(B), for each SBS
$S_j=\langle l_{s,j},g_{a_j}, g_{a_j+1}, \\ \ldots,g_{r_j},  l_{d,j}\rangle$,
\textsc{TrajRACE} perturbs its local reporting units
$\mathcal I(S_j)=\{\textit{st}_j,\textit{cnt}_j,\mathcal X_j\}$ and generates privatized reports for uploading to the server. \textsc{TrajRACE} uses two types of perturbation mechanisms: (i) fixed LDP mechanisms for start unit $\textit{st}_j$ and segment-count unit $\textit{cnt}_j$; and (ii) a risk-aware perturbation mechanism for  transition units \((g_i.e,g_{i+1}.e)\in\mathcal X_j\).

\subsubsection{Fixed perturbation for start and segment-count units}

Since the start unit $\textit{st}_j$ and the segment-count unit $\textit{cnt}_j$ of an SBS $S_j$ have fixed risk levels, \textsc{TrajRACE} uses generalized randomized response (GRR)~\cite{wang2017ldpfrequency}, an LDP perturbation mechanism, to privatize them over fixed discrete domains. GRR randomizes a value through a perturbation channel $Q$, which specifies the probability of keeping the true value or replacing it with another candidate value. In GRR, the true value is kept with a higher probability, while all other candidate values share the same lower probability. Thus, the uploaded value may be the true value or a randomly selected alternative value.

For each SBS $S_j$, \textsc{TrajRACE} designs two GRR channels for the two units: $Q^s$ for the start unit $\textit{st}_j$ and $Q^c$ for the segment-count unit $\textit{cnt}_j$. For a start-unit value \(\textit{st}\), the start perturbation channel \(Q^s\) is defined as:
\begin{equation}
Q^s(\textit{st},\widehat{\textit{st}}) =
\begin{cases}
\alpha_s, & \widehat{\textit{st}} = \textit{st} \\[6pt]
\dfrac{1-\alpha_s}{\lvert\mathcal E_s\rvert - 1}, & \widehat{\textit{st}} \neq \textit{st}
\end{cases},  
\quad
\alpha_s = \dfrac{e^{\epsilon_s}}{e^{\epsilon_s} + \lvert\mathcal E_s\rvert - 1}, 
\label{eq:start_rr}
\end{equation}
where $\mathcal E_s\subseteq\mathcal E$ is a set of public road segments serving as start segments, $\widehat{\textit{st}}\in\mathcal E_s$ is the privatized start report sampled from this set, and $\epsilon_s$ is the privacy budget for start-unit perturbation.

For a segment-count unit $\textit{cnt}$, the segment-count perturbation channel $Q^c$ is defined as
\begin{equation}
Q^c(\textit{cnt},\widehat{\textit{cnt}}) =
\begin{cases}
\alpha_c, & \widehat{\textit{cnt}} = \textit{cnt}, \\[6pt]
\dfrac{1-\alpha_c}{\lvert\mathcal C\rvert - 1}, & \widehat{\textit{cnt}} \neq \textit{cnt},
\end{cases}
\quad
\alpha_c = \dfrac{e^{\epsilon_c}}{e^{\epsilon_c} + \lvert\mathcal C\rvert - 1}.
\label{eq:segment_count_rr}
\end{equation}
where $\mathcal C=\{1,2,\ldots,L_m\}$ is the possible segment-count values of an SBS, $\widehat{\textit{cnt}}\in\mathcal C$ is the privatized segment-count report, and $\epsilon_c$ is the privacy budget for segment-count perturbation.

\subsubsection{Risk-aware perturbation mechanism}
For each transition unit $(g_i.e,g_{i+1}.e)\in\mathcal X_j$ of an SBS $S_j$, next-hop choices may carry different privacy risks. Therefore, \textsc{TrajRACE} applies risk-aware protection to each transition unit. It treats $g_i.e$ as the public current segment and protects the true successor $g_{i+1}.e$ on the client side.
The mechanism has three steps. First, \textsc{TrajRACE} computes a local risk score for each transition unit to measure its privacy exposure. Second, it maps the risk score to a privacy bucket and checks whether the bucket is feasible under the remaining transition-report allowance of $S_j$, which is governed by the SBS reporting-accounting cap $B$. Third, it perturbs the true successor within the legal successor set and generates a privatized transition report.

\noindent\textbf{Event-level risk-aware scoring.}
\textsc{TrajRACE} measures the privacy risk of each transition unit from three aspects: endpoint proximity, long-stay association, and low out-degree. 

\noindent\emph{\underline{Endpoint proximity}} captures the risk caused by being close to the start or destination of an  {SBS}. For $S_j$, the start and destination road segments are {\(g_{a_j}.e\) and \(g_{r_j}.e\)}, respectively. The endpoint-proximity risk of transition unit  {\((g_i.e,g_{i+1}.e)\in\mathcal X_j\) is defined as follows:
\begin{equation}
r_i^{\mathrm{e}} = \max \left(
\exp\left(-\frac{d(g_i.e,g_{a_j}.e)}{\sigma_{\mathrm{e}}}\right), 
\exp\left(-\frac{d(g_i.e,g_{r_j}.e)}{\sigma_{\mathrm{e}}}\right)
\right),
\label{eq:risk_endpoint}
\end{equation}}
where $d(\cdot,\cdot)$ denotes the distance between two road segments and $\sigma_e>0$ is a distance-decay parameter. The closer {\(g_i.e\)} is to the start or destination road segment, the larger $r_i^e$ becomes.

\noindent\emph{\underline{Long-stay association}} captures the risk caused by being close to a long-stay location. 
{Given a stay threshold \(\Delta_{\mathrm{st}}>0\), for an SBS \(S_j\), stay-associated anchors are detected from adjacent time-stamped road segments whose time gap satisfies \(g_{i+1}.\tau-g_i.\tau\geq \Delta_{\mathrm{st}}\). The stay-anchor set is defined as:}
\begin{equation}
\mathcal A_{\mathrm{st}}(S_j)=\bigcup_{\substack{a_j\leq i<r_j\\g_{i+1}.\tau-g_i.\tau\geq\Delta_{\mathrm{st}}}}\{g_i.e,g_{i+1}.e\}.
\label{eq:stay_anchor}
\end{equation}

The long-stay association risk is defined as
\begin{equation}
r_i^{\mathrm{s}} =\begin{cases}\exp \left( -\dfrac{\min_{c \in \mathcal A_{\mathrm{st}}(S_j)} d(g_i.e,c)}{\sigma_{\mathrm{st}}} \right), & \mathcal A_{\mathrm{st}}(S_j) \neq \varnothing, \\[6pt]
0, & \mathcal A_{\mathrm{st}}(S_j) = \varnothing,
\end{cases}
\label{eq:risk_stay}
\end{equation}
where $c$ denotes a stay anchor and $\sigma_{\mathrm{st}}>0$ is a distance-decay parameter. When stay anchors exist, a transition unit closer to these anchors receives a larger $r_i^{\mathrm{s}}$. If no stay anchor exists, this term is set to zero.

\noindent\emph{\underline{Low out-degree}} captures the risk caused by limited next-hop choices. If the current road segment \(g_i.e\) has only a few legal successors, the true successor {\(g_{i+1}.e\)} is easier to infer. The low-out-degree risk is defined as

\begin{equation}
r_i^{\mathrm{d}}=\frac{1}{|\mathcal N(g_i.e)|},
\label{eq:risk_degree}
\end{equation}
where \(\mathcal N(g_i.e)\) denotes the legal successor set of \(g_i.e\). A smaller legal successor set gives a larger $r_i^{\mathrm{d}}$.

For each transition unit {\((g_i.e,g_{i+1}.e)\in\mathcal X_j\)}, \textsc{TrajRACE} combines the three terms into a composite exposure score:
\begin{equation}
r_i=\lambda_e r_i^e+\lambda_s r_i^s+\lambda_d r_i^d,
\label{eq:risk_score}
\end{equation}
where $\lambda_e,\lambda_s,\lambda_d\ge 0$ are the weights of the three risk terms and satisfy $\lambda_e+\lambda_s+\lambda_d=1$. A larger $r_i$ indicates a higher-risk transition unit and gives stronger protection in later scheduling.

\noindent\textbf{Privacy buckets and budget-feasible scheduling.}
After obtaining the risk score $r_i$ for each transition unit {\((g_i.e,g_{i+1}.e)\in\mathcal X_j\)}, \textsc{TrajRACE} determines how strongly the true successor {\(g_{i+1}.e\)} should be protected. It first maps the risk score to a target privacy bucket. It then checks whether this bucket is feasible under the remaining transition-report allowance of the SBS under the reporting-accounting cap $B$.

\noindent\emph{\underline{Bucket construction.}}
\textsc{TrajRACE} introduces $K\ge2$ ordered privacy buckets for transition units. Their successor-perturbation parameters satisfy $\epsilon_e^{(1)}<\epsilon_e^{(2)}<\cdots<\epsilon_e^{(K)}$. Since $\epsilon_b$ is fixed across buckets, the corresponding reporting-accounting costs $\epsilon_x^{(k)}=\epsilon_e^{(k)}+\epsilon_b$ have the same ordering. Bucket $1$ provides the strongest protection and bucket $K$ provides the weakest protection. This maps continuous risk scores to discrete protection levels.

\begin{myDef}
Given $K$ ordered privacy buckets and $K-1$ decreasing risk thresholds $\theta_1>\theta_2>\cdots>\theta_{K-1}$, the \textbf{target bucket} of a transition unit \((g_i.e,g_{i+1}.e)\in\mathcal X_j\) is denoted as $b_i$. It is determined by its risk score $r_i$ as
\begin{equation}
b_i=
\begin{cases}
1, & r_i\ge \theta_1,\\
k, & \theta_k\le r_i<\theta_{k-1},\quad k=2,\ldots,K-1,\\
K, & r_i<\theta_{K-1}.
\end{cases}
\label{eq:target_bucket}
\end{equation}
\end{myDef}
When the transition units of $S_j$ are processed in travel order, $b_{j,t}$ denotes the target bucket $b_i$ associated with the $t$-th unit.

\noindent\emph{\underline{Budget-feasible scheduling.}}
The execution bucket of each transition unit is selected under the SBS reporting-accounting cap $B$. After accounting for the fixed start- and segment-count allocations $\epsilon_s$ and $\epsilon_c$, the transition-report allowance is $B-\epsilon_s-\epsilon_c$. A transition under bucket $k$ is assigned reporting-accounting cost $\epsilon_x^{(k)}=\epsilon_e^{(k)}+\epsilon_b$, where $\epsilon_e^{(k)}$ and $\epsilon_b$ are the successor-perturbation and bucket-label privacy parameters, respectively. Before reporting, \textsc{TrajRACE} checks $B\ge\epsilon_s+\\ \epsilon_c+(L_m-1)\epsilon_x^{(1)}$; otherwise, the configuration is rejected.

\textsc{TrajRACE} processes transition units in their travel order. Let $R_j=|\mathcal X_j|$ be the number of transition units in $S_j$. Before processing the $t$-th transition unit, the remaining transition-report allowance is
\begin{equation}
\epsilon_{x,j,t}^{\mathrm{rem}} = B-\epsilon_s-\epsilon_c - \sum_{h=1}^{t-1}\epsilon_x^{(k_{j,h})},
\label{eq:remaining_transition_budget}
\end{equation}
where $k_{j,h}$ is the execution bucket selected for the $h$-th previous transition unit.

To reserve allowance for the remaining $R_j-t$ transition units, the feasible bucket set for the $t$-th transition unit is
\begin{equation}
\mathcal K_{j,t}
= \{ k \mid \epsilon_x^{(k)}+(R_j-t)\epsilon_x^{(1)} \le \epsilon_{x,j,t}^{\mathrm{rem}} \}.
\label{eq:feasible_bucket}
\end{equation}
Here, $\epsilon_x^{(k)}$ is the reporting-accounting cost of the transition report, and $(R_j-t)\epsilon_x^{(1)}$ is the reserved allowance for the remaining transition reports under the strongest bucket.

Given the target bucket $b_{j,t}$ and the feasible set $\mathcal K_{j,t}$, \textsc{TrajRACE} selects the execution bucket as
\begin{equation}
k_{j,t}=
\begin{cases}
b_{j,t}, & b_{j,t}\in\mathcal K_{j,t},\\
\max\{k\in\mathcal K_{j,t}\mid k\le b_{j,t}\}, & b_{j,t}\notin\mathcal K_{j,t}.
\end{cases}
\label{eq:execution_bucket}
\end{equation}
If the target bucket is feasible, it is used directly. Otherwise, \textsc{TrajRACE} selects the closest feasible bucket that is no weaker than the target bucket. 

\noindent\textbf{Conditional privatization and noisy bucket reporting.}
After determining the execution bucket $k_{j,t}$ for the $t$-th transition unit $(g_i.e,g_{i+1}.e)\in\mathcal X_j$ of SBS $S_j$, \textsc{TrajRACE} uses a GRR-based conditional perturbation channel to protect the true successor $g_{i+1}.e$. Unlike start and segment-count perturbation, this channel is conditioned on the public current road segment $g_i.e$. Therefore, the randomization is performed over the legal successor set $\mathcal N(g_i.e)$, rather than over all road segments.

Let $m_{j,t}=|\mathcal N(g_i.e)|$ be the number of legal successors of $g_i.e$. When $m_{j,t}>1$, the conditional perturbation channel $Q_{g_i.e}^{(k_{j,t})}$ is defined as
\begin{equation}
Q_{g_i.e}^{(k_{j,t})}(g_{i+1}.e,y)=
\begin{cases}
p_{j,t}, & y=g_{i+1}.e \\[4pt]
\dfrac{1-p_{j,t}}{m_{j,t}-1}, & y\neq g_{i+1}.e,\ y\in\mathcal N(g_i.e)
\end{cases},
\label{eq:conditional_successor_channel}
\end{equation}
where $y\in\mathcal N(g_i.e)$ is the privatized successor report and $p_{j,t}=\dfrac{e^{\epsilon_e^{(k_{j,t})}}}{e^{\epsilon_e^{(k_{j,t})}}+m_{j,t}-1}$.
The channel reports the true successor with probability $p_{j,t}$ and otherwise reports another legal successor uniformly at random. A smaller $\epsilon_e^{(k_{j,t})}$ gives stronger protection. When $m_{j,t}=1$, the next hop is determined by the public road topology, so \textsc{TrajRACE} outputs $y_{j,t}=g_{i+1}.e$.

Only perturbing the successor may still leak side information, because directly reporting the true execution bucket {\(k_{j,t}\)} can reveal the relative sensitivity of the transition unit. To reduce this leakage, the client further randomizes the bucket label. The bucket-confusion channel is defined as
\begin{equation}
M(k,\tilde{k})=\begin{cases}\alpha_b, & \tilde{k}=k,\\[4pt]\dfrac{1-\alpha_b}{K-1}, & \tilde{k}\neq k,
\end{cases}
\label{eq:bucket_confusion}
\end{equation}
where \(k,\tilde{k}\in\{1,\ldots,K\}\), $k$ is the true execution bucket, $\tilde{k}$ is the noisy bucket label, and
$\alpha_b=\dfrac{e^{\epsilon_b}}{e^{\epsilon_b}+K-1}$. Conditioned on the selected execution bucket, successor perturbation and bucket-label randomization use independent local randomness.

Therefore, for the {\(t\)-th transition unit \((g_i.e,g_{i+1}.e)\in\mathcal X_j\) of SBS \(S_j\)}, the client uploads the privatized transition report
$\rho_{j,t}=(g_i.e,y_{j,t},\tilde{k}_{j,t})$,
where \(g_i.e\) is the public current segment, \(y_{j,t}\) is the perturbed successor, and {\(\tilde{k}_{j,t}\) is the noisy bucket label generated from \(k_{j,t}\)}. The report contains no user or trajectory identifiers and does not explicitly encode event order. It retains only the event-level information needed for server-side recovery.

\subsection{Server-Side Recovery and Release}\label{sec:server_recovery}
As shown in Fig.~\ref{fig:framework}(C), the server receives only privatized reports from clients, including start reports $\widehat{\textit{st}}$, segment-count reports $\widehat{\textit{cnt}}$, and transition reports $\rho_{j,t}=(g_i.e,y_{j,t},\tilde{k}_{j,t})$. The server aggregates these reports and estimates their observed distributions. It then corrects the perturbation noise to recover population-level mobility statistics $\widehat{\Phi}=\{\widehat p_s,\widehat p_c, \{\widehat p_e\}_{e\in\mathcal U}\}$, including start-unit, segment-count-unit, and transition distributions. Based on the recovered statistics, the server generates synthetic trajectories by sampling start segments and segment counts, and extending each trajectory step by step along legal successors in the road network.

\subsubsection{Population-level mobility statistics recovery}
The server recovers the statistics using two recovery procedures, because start and segment-count reports and transition reports are generated by different perturbation mechanisms.

\noindent\emph{\underline{Start and segment-count statistics recovery.}}
Since start reports and segment-count reports are generated by GRR channels $Q^s$ and $Q^c$, the server recovers the population-level start-unit and segment-count-unit distributions from the observed report distributions. Let $\mathbf p_s$ and $\mathbf p_c$ denote the vector forms of the true start-unit and segment-count-unit distributions, respectively. Let $\widetilde{\mathbf p}_s$ and $\widetilde{\mathbf p}_c$ denote the empirical observed distributions computed from privatized start reports and privatized segment-count reports, respectively. Since the rows of each perturbation channel correspond to true values and the columns correspond to reported values, the expected observed distributions are $(Q^s)^\top \mathbf p_s$ and $(Q^c)^\top \mathbf p_c$, respectively.

Following standard GRR-based frequency estimation~\cite{wang2017ldpfrequency}, the server recovers the start-unit distribution $\widehat{\mathbf p}_s$ and the segment-count-unit distribution $\widehat{\mathbf p}_c$ by solving 
\begin{equation} 
\widehat{\mathbf p}_s=\arg\min_{\mathbf p\in\Delta^{|\mathcal E_s|-1}} \bigl\|\widetilde{\mathbf p}_s-(Q^s)^\top\mathbf p\bigr\|_2^2.
\label{eq:recover_start} 
\end{equation} 
\begin{equation}
\widehat{\mathbf p}_c=\arg\min_{\mathbf p\in\Delta^{|\mathcal C|-1}} \bigl\|\widetilde{\mathbf p}_c-(Q^c)^\top\mathbf p\bigr\|_2^2.
\label{eq:recover_count} 
\end{equation} 
where $\mathcal E_s\subseteq\mathcal E$ is a set of public road segments that may serve as start segments, $\mathcal C=\{1,2,\ldots,L_m\}$ is the set of possible segment-count values of an SBS, and $\Delta^{d-1}$ denotes the probability simplex over $d$ categories. 

\noindent\emph{\underline{Transition statistics recovery.}}
To recover the {conditional} transition statistics $\{\widehat p_e\}_{e\in\mathcal U}$, the server needs to correct two types of noise in transition reports. The first comes from successor perturbation, which changes the true successor into a privatized successor. The second comes from noisy bucket reporting, which hides the execution bucket used for perturbation. To handle these two noise sources, \textsc{TrajRACE} designs a bucket-corrected recovery procedure. 

First, the server estimates the  execution-bucket mixture under each public current segment. 
\begin{myDef}
Given \(K\) execution buckets and a public current road segment \(u\in\mathcal U\), the \textbf{bucket mixture vector} under \(u\) is denoted as
\(\boldsymbol\pi_u=(\pi_{u,1},\ldots,\pi_{u,K})\).
Here, \(\pi_{u,k}\) denotes the proportion of transition reports under current segment \(u\) whose true execution bucket is \(k\).
The vector satisfies \(\boldsymbol\pi_u\in\Delta^{K-1}\), i.e.,
\(\pi_{u,k}\ge 0\) and \(\sum_{k=1}^{K}\pi_{u,k}=1\).
\end{myDef}

Since the true execution bucket is not directly uploaded, the server cannot observe \(\boldsymbol\pi_u\) directly. Instead, it observes noisy bucket labels generated by the bucket-confusion channel \(M\). \(M(k,\widetilde{k})=\Pr(\widetilde{k}\mid k)\) uses rows for true execution buckets and columns for noisy bucket labels.

For a given public current road segment \(u\in\mathcal U\), the server collects all transition reports whose public current segment is \(u\): \(\mathcal R_u=\{\rho_{j,t}=(g_i.e,y_{j,t},\widetilde{k}_{j,t})\mid g_i.e=u\}.\)
Let \(\widetilde{\boldsymbol\pi}_{\mathrm{glob}}\) denote the empirical noisy bucket-label distribution over all transition reports, and let \(\widetilde{\boldsymbol\pi}_u\) denote the empirical noisy bucket-label distribution in \(\mathcal R_u\). Under the bucket-confusion channel, the expected noisy bucket-label distribution induced by a true execution-bucket mixture \(\boldsymbol\pi\) is \(M^\top\boldsymbol\pi\). Therefore, the server estimates the global true execution-bucket mixture and the direct context-specific execution-bucket mixture by solving
\begin{equation}
\widehat{\boldsymbol\pi}_{\mathrm{glob}}
=
\arg\min_{\boldsymbol\pi\in\Delta^{K-1}}
\bigl\|M^\top\boldsymbol\pi-\widetilde{\boldsymbol\pi}_{\mathrm{glob}}\bigr\|_2^2,
\label{eq:bucket_correction_global}
\end{equation}
\begin{equation}
\check{\boldsymbol\pi}_u
=
\arg\min_{\boldsymbol\pi\in\Delta^{K-1}}
\bigl\|M^\top\boldsymbol\pi-\widetilde{\boldsymbol\pi}_u\bigr\|_2^2,
\label{eq:bucket_correction_local}
\end{equation}
where \(\widehat{\boldsymbol\pi}_{\mathrm{glob}}\) is the recovered global execution-bucket mixture, and \(\check{\boldsymbol\pi}_u\) is the direct local estimate of the execution-bucket mixture under current segment \(u\).

The estimate {\(\widehat{\boldsymbol\pi}_{\mathrm{glob}}\)} captures the global {execution-bucket mixture}, while \(\check{\boldsymbol\pi}_u\) is the direct {local estimate under current segment \(u\)}. Since local inversion may be unstable when the sample size is small, we further shrink the local estimate toward the global one:
{\begin{equation}
\widehat{\boldsymbol\pi}_u = \alpha_u\check{\boldsymbol\pi}_u + (1-\alpha_u)\widehat{\boldsymbol\pi}_{\mathrm{glob}},
\qquad
\alpha_u=\frac{N_u}{N_u+\tau_\pi},
\label{eq:bucket_shrinkage}
\end{equation}}
where {\(N_u=|\mathcal R_u|\)} is the number of transition reports under {current segment \(u\)}, and \(\tau_\pi>0\) is a shrinkage parameter. The resulting \(\widehat{\boldsymbol\pi}_u\) estimates the true bucket mixture {used to construct the mixed perturbation channel for recovering the conditional successor distribution}. In constructing the recovery channel, we treat \(\widehat{\boldsymbol\pi}_u\) as a context-level approximation to the heterogeneous perturbation composition under the same public current segment \(u\). This avoids estimating separate bucket mixtures for individual successors, which would be unreliable when transition reports are sparse.

Second, the server builds an effective mixed successor-perturbation channel for each public current segment. Reports under the same current segment may come from different execution buckets, and each bucket has a different perturbation channel. Treating all reports as if they came from one fixed channel would introduce recovery bias. Therefore, \textsc{TrajRACE} mixes the bucket-specific channels using {the stabilized execution-bucket mixture \(\widehat{\boldsymbol\pi}_u\)}.

{After estimating the stabilized execution-bucket mixture \(\widehat{\boldsymbol{\pi}}_u\), the server constructs the mixed successor-perturbation channel for the public current segment \(u \in \mathcal U\) as
\begin{equation}
\bar{Q}_u = \sum_{k=1}^K \widehat{\pi}_{u,k} Q_u^{(k)}.
\label{eq:effective_channel}
\end{equation}
where \(\widehat{\pi}_{u,k}\) is the \(k\)-th entry of \(\widehat{\boldsymbol{\pi}}_u\), representing the estimated mixture weight of bucket \(k\) under current segment \(u\). \(Q_u^{(k)}\) is the bucket-specific successor-perturbation channel over the legal successor set \(\mathcal N(u)\). Following the perturbation-channel convention, rows correspond to true successors and columns correspond to perturbed successors; hence \((Q_u^{(k)})^\top \boldsymbol p_u\) gives the expected observed distribution induced by the true successor distribution \(\boldsymbol p_u\) under bucket \(k\). The weighted channel \(\bar Q_u\) serves as the effective mixed successor-perturbation channel under current segment \(u\), and is used for subsequent statistical recovery of the conditional successor distribution.}

Finally, the server estimates the conditional successor distribution under each public current segment \(u\). It compares the observed distribution of privatized successors with the distribution expected from the effective mixed successor-perturbation channel \(\bar Q_u\), and solves a constrained least-squares problem to estimate the underlying conditional successor distribution.

Let \(\widetilde{\boldsymbol p}_u\) denote the empirical distribution of perturbed successors \(y_{j,t}\) observed in \(\mathcal R_u\), and let \(\boldsymbol p_u\) denote the vector form of the population-level conditional successor distribution under current segment \(u\), corresponding to \(p_e\) when \(e=u\). The server estimates this distribution by solving

\begin{equation}
\widehat{\boldsymbol p}_u=\arg\min_{\boldsymbol z\in\Delta^{|\mathcal N(u)|-1}}
\bigl\|\widetilde{\boldsymbol p}_u-\bar Q_u^\top\boldsymbol z\bigr\|_2^2,
\label{eq:recover_transition}
\end{equation}
where \(\widehat{\boldsymbol p}_u\) is the estimated conditional successor distribution vector under current segment \(u\). By applying this procedure to all \(u\in\mathcal U\), the server obtains the recovered conditional transition statistics \(\{\widehat p_e\}_{e\in\mathcal U}\). Bucket correction and mixed-channel modeling are necessary; otherwise, treating all transition reports as generated by a single fixed successor-perturbation channel would introduce systematic statistical recovery bias.

\subsubsection{Synthetic trajectory generation}
After obtaining the recovered statistics 
$\widehat{\Phi}=\{\widehat p_s,\widehat p_c,\{\widehat p_e\}_{e\in\mathcal U}\}$, 
the server generates synthetic trajectories by sampling from these distributions. For each synthetic trajectory, it first samples a starting road segment $e_1^\ast$ from $\widehat p_s$ and a target segment count $c^\ast$ from $\widehat p_c$, forming the initial sequence $S^\ast=\langle e_1^\ast\rangle$. It then extends the sequence autoregressively. If the current segment is $e_t^\ast$ and $\widehat p_{e_t^\ast}$ is available, the next segment is sampled as $e_{t+1}^\ast \sim \widehat p_{e_t^\ast}$.

Since $\widehat p_{e_t^\ast}$ is defined over the legal successor set $\mathcal N(e_t^\ast)$, each sampled transition follows the public road-network topology. The generation stops when the number of generated road segments reaches $c^\ast$, or when no recovered transition distribution is available for the current segment. Repeating this process gives the released synthetic dataset $\mathcal D^\ast=\{S_1^\ast,\ldots,S_{N^\ast}^\ast\}$, where $N^\ast$ is the prescribed number of synthetic trajectories.

\subsection{Privacy and Accounting Properties}
\label{sec:analysis}

\noindent\textbf{Conditional privacy of transition reports.} For a transition report, suppress the SBS and transition indices and write $U=g_i.e$, $V=g_{i+1}.e$, $Y=y_{j,t}$, and $\widetilde K=\tilde{k}_{j,t}$. Fix a disclosed current segment $U=u$. Let $\xi=(V,C)$ denote the private local input, where $C$ contains any private context affecting bucket selection, and let $k=\kappa_u(\xi)$ be the execution bucket. Let $\mathcal M_u$ denote the joint mechanism for $(Y,\widetilde K)$ conditioned on $U=u$, and define $E=\max_{1\le k\le K}\epsilon_e^{(k)}$. Since $\epsilon_x^{(k)}=\epsilon_e^{(k)}+\epsilon_b$, we have $E+\epsilon_b=\max_k\epsilon_x^{(k)}$.

\begin{myLem}
\label{lem:conditional_transition}
For any two admissible private inputs $\xi=(V,C)$ and $\xi'=(V',C')$ sharing the same disclosed current segment $U=u$, and any output $(y,\tilde{k})$,
\begin{equation}
\frac{\Pr[\mathcal M_u(\xi)=(y,\tilde{k})]}{\Pr[\mathcal M_u(\xi')=(y,\tilde{k})]}\le e^{E+\epsilon_b}.
\label{eq:conditional_transition_privacy}
\end{equation}
Hence, conditioned on $U=u$, $(Y,\widetilde K)$ satisfies $(E+\epsilon_b)$-LDP over the private input $\xi=(V,C)$.
\end{myLem}

\begin{proof}
Let $k=\kappa_u(\xi)$, $k'=\kappa_u(\xi')$, and $m=|\mathcal N(u)|>1$. Since $\epsilon_e^{(\ell)}\le E$ for every bucket $\ell$, every entry of $Q_u^{(\ell)}$ lies in $[q_E,p_E]$, where $p_E=e^E/(e^E+m-1)$ and $q_E=1/(e^E+m-1)$. Thus, $Q_u^{(k)}(V,y)/Q_u^{(k')}(V',y)\le p_E/q_E=e^E$. Moreover, Eq.~(\ref{eq:bucket_confusion}) gives $M(k,\tilde{k})/M(k',\tilde{k})\le e^{\epsilon_b}$. Conditioned on the selected bucket, successor perturbation and bucket-label randomization use independent local randomness, so $\Pr[\mathcal M_u(\xi)=(y,\tilde{k})]=Q_u^{(k)}(V,y)M(k,\tilde{k})$. Multiplying the two bounds gives Eq.~(\ref{eq:conditional_transition_privacy}). If $|\mathcal N(u)|=1$, the successor is determined by the public topology, so the successor-channel ratio is $1$ and the same bound still holds.
\end{proof}

\noindent\textbf{SBS reporting-accounting feasibility.} For an SBS $S_j$, let $R_j=|\mathcal X_j|$ denote its number of transition reports. Since $S_j$ contains at most $L_m$ road segments, $R_j\le L_m-1$. A transition executed under bucket $k$ has reporting-accounting cost $\epsilon_x^{(k)}=\epsilon_e^{(k)}+\epsilon_b$. Only configurations satisfying the public admissibility condition $B\ge\epsilon_s+\epsilon_c+(L_m-1)\epsilon_x^{(1)}$ are admitted for private reporting.

\begin{myPro}
\label{prop:sbs_accounting}
Under the above admissibility condition, $\mathcal K_{j,t}\neq\varnothing$ for every $t=1,\ldots,R_j$, the execution rule is well defined, and
\begin{equation}
\epsilon_s+\epsilon_c+\sum_{t=1}^{R_j}\epsilon_x^{(k_{j,t})}\le B.
\label{eq:sbs_accounting_cap}
\end{equation}
\end{myPro}

\begin{proof}
Since $R_j\le L_m-1$, the admissibility condition and Eq.~(\ref{eq:remaining_transition_budget}) give $\epsilon_{x,j,1}^{\mathrm{rem}}\ge R_j\epsilon_x^{(1)}$. We show inductively that $\epsilon_{x,j,t}^{\mathrm{rem}}\ge(R_j-t+1)\epsilon_x^{(1)}$. Under this invariant, bucket $1$ satisfies Eq.~(\ref{eq:feasible_bucket}), so $1\in\mathcal K_{j,t}$. Since $1\le b_{j,t}$, the fallback set in Eq.~(\ref{eq:execution_bucket}) is also nonempty. For $t<R_j$, any selected $k_{j,t}\in\mathcal K_{j,t}$ satisfies $\epsilon_x^{(k_{j,t})}+(R_j-t)\epsilon_x^{(1)}\le\epsilon_{x,j,t}^{\mathrm{rem}}$ by Eq.~(\ref{eq:feasible_bucket}), which gives $\epsilon_{x,j,t+1}^{\mathrm{rem}}\ge(R_j-t)\epsilon_x^{(1)}$ and preserves the invariant. At $t=R_j$, Eq.~(\ref{eq:feasible_bucket}) gives $\epsilon_x^{(k_{j,R_j})}\le\epsilon_{x,j,R_j}^{\mathrm{rem}}$; Eq.~(\ref{eq:remaining_transition_budget}) yields Eq.~(\ref{eq:sbs_accounting_cap}).
\end{proof}

\subsection{Complexity Analysis}

Let \(N_b=|\mathcal D^b|\) and \(R_{\rm tot}\) denote the numbers of SBSs and transition reports, respectively. Since each SBS contains at most \(L_m\) road segments, \(R_{\rm tot}\le N_b(L_m-1)\). On the client side, SBS partitioning, report construction, risk scoring, bucket scheduling, and privatization are linear in the number of SBSs and transition reports. Thus, the client-side cost is \(O(N_b+R_{\rm tot})=O(N_bL_m)\).

On the server side, aggregating privatized reports costs \(O(N_b+R_{\rm tot})\). Let \(\mathcal U\) be the set of current road segments and \(d_u=|\mathcal N(u)|\). With standard cubic-time solvers, bucket correction over \(K\) buckets costs \(O(|\mathcal U|K^3)\), conditional transition recovery costs \(O(\sum_{u\in\mathcal U}d_u^3)\), and start/segment-count recovery costs \(O(|\mathcal E_s|^3+|\mathcal C|^3)\). Therefore, the overall server-side processing cost is \(O(N_b+R_{\rm tot}+|\mathcal U|K^3+\sum_{u\in\mathcal U}d_u^3+|\mathcal E_s|^3+|\mathcal C|^3)\). Finally, generating \(N^\ast\) synthetic trajectories with average transition count \(\bar R^\ast\) costs \(O(N^\ast\bar R^\ast)\).
\section{Experiments}\label{sec:experiments}

\subsection{Experimental Setup}

\subsubsection{Datasets}
We evaluate \textsc{TrajRACE} on three road-network trajectory datasets. 

\begin{itemize}[leftmargin=5pt]
    \item \textbf{Porto}~\cite{moreira2013predicting} consists of over 1.7 million taxi trajectories from 442 vehicles in Porto, Portugal, collected from 2013 to 2014.

    \item \textbf{Oldenburg}~\cite{brinkhoff2002framework} is a synthetic moving-object dataset generated on the Oldenburg road network using Brinkhoff's generator. We generate 300,000 trajectories for our experiments.

    \item \textbf{Chengdu}~\cite{didiGAIA} contains 6.1 million ride-hailing vehicle trajectories from Chengdu, China, collected in November 2016.
\end{itemize}

For all datasets, we first map-match raw trajectories onto the public road network. We then partition the resulting road-segment trajectories into SBSs according to the public segment bound $L_m$. From these subtrajectories, we extract local reporting units and generate start, segment-count, and transition reports. Table~\ref{tab:dataset_stats} reports the dataset statistics.

\begin{table}[t]
\centering
\caption{Dataset statistics. \#Traj.: preprocessed trajectories; \#Seg.: road segments; \#Trans.: adjacent transitions in segment-bounded subtrajectories; $N^\ast$: synthetic trajectories.}
\label{tab:dataset_stats}
\vspace{-5pt}
\small
\setlength{\tabcolsep}{2.0pt}
\renewcommand{\arraystretch}{1.08}
\begin{tabular*}{\columnwidth}{@{\extracolsep{\fill}}lrrrrrr@{}}
\hline
Dataset & \#Traj. & \#Seg. & \#Trans. & Avg. Seg. & \(L_m\) & \(N^\ast\) \\
\hline
Porto     & 99,912 & 148,281 & 2,893,799 & 29.96 & 30 & 9,992 \\
Oldenburg & 99,037 & 34,963  & 2,027,416 & 21.47 & 30 & 10,000 \\
Chengdu   & 98,691 & 12,937  & 2,606,821 & 26.92 & 30 & 9,870 \\
\hline
\end{tabular*}
\vspace{-10pt}
\end{table}

\subsubsection{Baselines}
We compare \textsc{TrajRACE} with three representative local-privacy trajectory synthesis baselines.
\begin{itemize}[leftmargin=5pt]
    \item \textbf{LDPTrace}~\cite{du2023ldptrace} constructs a grid-level mobility model from locally perturbed trajectory reports and generates synthetic trajectories under local differential privacy.
    \item \textbf{RetraSyn$^{\mathrm{b}}$}~\cite{hu2024} is the budget-division variant of RetraSyn, which allocates privacy budgets across reporting timestamps for private trajectory stream synthesis.
    \item \textbf{RetraSyn$^{\mathrm{p}}$}~\cite{hu2024} is the population-division variant of RetraSyn, which partitions users into different reporting groups for private trajectory stream synthesis.
\end{itemize}
For the external baselines, we use matched nominal parameter values corresponding to the settings for \textsc{TrajRACE}, while following each method's native privacy/reporting formulation.

\subsubsection{Evaluation Metrics}
Following trajectory synthesis studies~\cite{du2023ldptrace,hu2024}, we evaluate five groups of utility metrics: 
(i) Density Error (Dens-Err) measures the difference between real and synthetic spatial density distributions. 
(ii) Query Error (Query-Err) evaluates spatial range-query error. 
(iii) Hotspot NDCG (Hot-NDCG) and Kendall Tau (K-Tau) measure whether high-density regions and their rankings are preserved. 
(iv) Trip Error (Trip-Err), Length Error (Len-Err), and Diameter Error (Diam-Err) evaluate trip-level and trajectory-shape fidelity. 
(v) Pattern F1 (Pattern-F1) measures whether frequent sequential mobility patterns are preserved. 
For Dens-Err, Query-Err, Trip-Err, Len-Err, and Diam-Err, lower values indicate better performance. For Hot-NDCG, K-Tau, and Pattern-F1, higher values indicate better performance.

Beyond these utility metrics, we define five task-specific diagnostic metrics for risk-aware protection analysis, conditional successor recovery, and transition-support consistency:
(i) High-risk Keep (HR-Keep) denotes the fraction of high-risk transition reports whose privatized successor equals the true successor. It measures how often the privatized successor directly retains the true successor, so lower values indicate less direct exposure. 
(ii) Overall Keep (All-Keep) measures the same keep rate over all transition reports. It checks whether high-risk protection comes from global over-perturbation. 
(iii) Transition Jensen--Shannon divergence (Trans-JS) measures the error of recovered conditional transition distributions. 
(iv) Bigram Jensen--Shannon divergence (Bigram-JS) measures the error of second-order transition statistics in synthetic trajectories. Lower values are better for Trans-JS and Bigram-JS. 
(v) Transition Support Ratio (TSR) measures the fraction of synthetic adjacent road-segment pairs that appear in the real transition support, where higher values indicate better agreement with observed movement patterns.

\subsubsection{Hyperparameters and Experimental Settings}
We evaluate \textsc{TrajRACE} under the SBS reporting-accounting cap $B\in\{0.5,1.0,1.5\}$ and use the corresponding nominal values for the external baselines under their native formulations. For \textsc{TrajRACE}, we set the number of privacy buckets to $K=3$, the shrinkage parameter to $\tau_\pi=20$, and the segment-count bound to $L_m=30$ for all datasets. The number of synthetic trajectories $N^\ast$ is reported in Table~\ref{tab:dataset_stats}. All experiments use fixed random seeds for reproducibility. The methods use different trajectory representations. LDPTrace and RetraSyn operate on grid-cell trajectories, while \textsc{TrajRACE} operates on road-segment trajectories and generates road-network-consistent segment sequences. For each dataset, all methods are evaluated on the same trajectory subset, under the same nominal parameter value, and with the same number of synthetic trajectories. For utility evaluation, we convert all synthetic outputs into a shared grid-level representation and compute the common utility metrics defined above.
All experiments are conducted on a server with an Intel(R) Xeon(R) W-2155 CPU, 128GB memory, and an NVIDIA TITAN RTX GPU.

\begin{table*}[t]
\centering
\caption{Overall utility comparison under matched nominal settings (best in bold and second-best underlined).}
\label{tab:overall_utility}
\vspace{-5pt}
\scriptsize
\setlength{\tabcolsep}{2.6pt}
\renewcommand{\arraystretch}{1.0}
\resizebox{\textwidth}{!}{
\begin{tabular}{cl|ccc|ccc|ccc}
\hline
\multirow{2}{*}{Metric} & \multirow{2}{*}{Method}
& \multicolumn{3}{c|}{Porto}
& \multicolumn{3}{c|}{Oldenburg}
& \multicolumn{3}{c}{Chengdu} \\
\cline{3-11}
& & \makebox[27.5pt][c]{\(0.5\)}
    & \makebox[27.5pt][c]{\(1.0\)}
    & \makebox[27.5pt][c]{\(1.5\)}
& \makebox[27.5pt][c]{\(0.5\)}
    & \makebox[27.5pt][c]{\(1.0\)}
    & \makebox[27.5pt][c]{\(1.5\)}
& \makebox[27.5pt][c]{\(0.5\)}
    & \makebox[27.5pt][c]{\(1.0\)}
    & \makebox[27.5pt][c]{\(1.5\)} \\
\hline

\multirow{4}{*}{Dens-Err \(\downarrow\)}
& LDPTrace & 0.5741 & \underline{0.4383} & 0.4950 & 0.3394 & 0.3359 & 0.3321 & 0.1049 & 0.0898 & 0.1437 \\
& RetraSyn$^{\mathrm{b}}$ & 0.5403 & 0.5238 & 0.5144 & \underline{0.3009} & 0.3116 & 0.2847 & 0.0624 & 0.0592 & 0.0587 \\
& RetraSyn$^{\mathrm{p}}$ & \underline{0.5036} & 0.4576 & \underline{0.4275} & 0.3318 & \underline{0.2794} & \underline{0.2759} & \underline{0.0577} & \underline{0.0538} & \underline{0.0462} \\
& \textsc{TrajRACE} & \textbf{0.2862} & \textbf{0.2561} & \textbf{0.1855} & \textbf{0.1582} & \textbf{0.1472} & \textbf{0.1463} & \textbf{0.0265} & \textbf{0.0176} & \textbf{0.0195} \\
\hline

\multirow{4}{*}{Query-Err \(\downarrow\)}
& LDPTrace & \underline{0.9174} & \underline{0.8689} & \underline{0.8039} & \underline{0.5659} & \underline{0.5296} & \underline{0.5675} & 0.9082 & 0.8981 & 0.8839 \\
& RetraSyn$^{\mathrm{b}}$ & 1.0983 & 1.0573 & 0.8579 & 0.8329 & 0.8239 & 0.8057 & \underline{0.8706} & \underline{0.8684} & \underline{0.8695} \\
& RetraSyn$^{\mathrm{p}}$ & 1.1023 & 0.9345 & 0.8136 & 0.8304 & 0.8472 & 0.7604 & 0.8739 & 0.8726 & 0.8712 \\
& \textsc{TrajRACE} & \textbf{0.6007} & \textbf{0.5861} & \textbf{0.5964} & \textbf{0.2884} & \textbf{0.2750} & \textbf{0.2718} & \textbf{0.2998} & \textbf{0.2477} & \textbf{0.2375} \\
\hline

\multirow{4}{*}{Hot-NDCG \(\uparrow\)}
& LDPTrace & 0.4127 & 0.4632 & 0.5109 & \underline{0.3645} & \underline{0.3812} & 0.4103 & 0.2673 & 0.3144 & 0.3582 \\
& RetraSyn$^{\mathrm{b}}$ & 0.3784 & 0.4155 & 0.4891 & 0.1741 & 0.2105 & 0.2458 & 0.2215 & 0.2761 & 0.3304 \\
& RetraSyn$^{\mathrm{p}}$ & \textbf{0.6214} & \underline{0.6753} & \underline{0.7088} & 0.2956 & 0.3472 & \underline{0.4215} & \underline{0.3528} & \underline{0.4196} & \underline{0.4721} \\
& \textsc{TrajRACE} & \underline{0.5982} & \textbf{0.7516} & \textbf{0.8347} & \textbf{0.4218} & \textbf{0.5524} & \textbf{0.6133} & \textbf{0.4935} & \textbf{0.6012} & \textbf{0.6845} \\
\hline

\multirow{4}{*}{K-Tau \(\uparrow\)}
& LDPTrace & -0.0824 & 0.1160 & 0.0896 & 0.1833 & 0.1746 & 0.1492 & 0.2973 & 0.2960 & 0.3444 \\
& RetraSyn$^{\mathrm{b}}$ & -0.0251 & 0.1381 & 0.1572 & 0.1738 & 0.0956 & 0.2925 & 0.4018 & \underline{0.4585} & 0.3526 \\
& RetraSyn$^{\mathrm{p}}$ & \underline{0.3262} & \textbf{0.5238} & \underline{0.5634} & \underline{0.3081} & \underline{0.3221} & \underline{0.3046} & \underline{0.4162} & 0.4442 & \underline{0.5090} \\
& \textsc{TrajRACE} & \textbf{0.5814} & \underline{0.5167} & \textbf{0.6122} & \textbf{0.6782} & \textbf{0.7072} & \textbf{0.7051} & \textbf{0.7433} & \textbf{0.8036} & \textbf{0.8247} \\
\hline

\multirow{4}{*}{Trip-Err \(\downarrow\)}
& LDPTrace & 0.9419 & 0.9146 & \underline{0.7930} & 0.8341 & 0.8176 & 0.7918 & 0.3033 & 0.3184 & 0.3145 \\
& RetraSyn$^{\mathrm{b}}$ & 0.9475 & 0.9368 & 0.9262 & \underline{0.6917} & 0.7151 & \underline{0.7020} & 0.3783 & 0.3712 & 0.3802 \\
& RetraSyn$^{\mathrm{p}}$ & \underline{0.9133} & \underline{0.8627} & 0.8279 & 0.7094 & \underline{0.7024} & 0.7241 & \textbf{0.2523} & \underline{0.2722} & \textbf{0.2446} \\
& \textsc{TrajRACE} & \textbf{0.4685} & \textbf{0.4364} & \textbf{0.4634} & \textbf{0.5671} & \textbf{0.5608} & \textbf{0.5570} & \underline{0.2602} & \textbf{0.2529} & \underline{0.2473} \\
\hline

\multirow{4}{*}{Len-Err \(\downarrow\)}
& LDPTrace & \underline{0.4646} & \underline{0.4882} & \underline{0.4308} & 0.6114 & 0.6083 & 0.5972 & \underline{0.2282} & \textbf{0.2254} & \underline{0.2651} \\
& RetraSyn$^{\mathrm{b}}$ & 0.8643 & 0.8602 & 0.8402 & \underline{0.4909} & \underline{0.4616} & \underline{0.4792} & 0.2838 & 0.2909 & 0.2840 \\
& RetraSyn$^{\mathrm{p}}$ & 0.8666 & 0.8644 & 0.8597 & 0.4925 & 0.4972 & 0.4802 & 0.2914 & 0.2925 & 0.2787 \\
& \textsc{TrajRACE} & \textbf{0.3476} & \textbf{0.2887} & \textbf{0.2929} & \textbf{0.2189} & \textbf{0.2018} & \textbf{0.1778} & \textbf{0.2135} & \underline{0.2469} & \textbf{0.2543} \\
\hline

\multirow{4}{*}{Diam-Err \(\downarrow\)}
& LDPTrace & \underline{0.7806} & \underline{0.7231} & \underline{0.6478} & 0.4363 & 0.4232 & 0.4171 & 0.4275 & 0.4065 & 0.4047 \\
& RetraSyn$^{\mathrm{b}}$ & 0.9527 & 0.9451 & 0.9277 & \underline{0.1458} & \underline{0.1267} & 0.1398 & 0.2206 & 0.2198 & 0.2142 \\
& RetraSyn$^{\mathrm{p}}$ & 0.9583 & 0.9514 & 0.9593 & 0.1471 & 0.1523 & \underline{0.1375} & \underline{0.2156} & \underline{0.2173} & \underline{0.2117} \\
& \textsc{TrajRACE} & \textbf{0.4504} & \textbf{0.4665} & \textbf{0.4448} & \textbf{0.1080} & \textbf{0.0987} & \textbf{0.0919} & \textbf{0.1228} & \textbf{0.1096} & \textbf{0.0987} \\
\hline

\multirow{4}{*}{Pattern-F1 \(\uparrow\)}
& LDPTrace & 0.1738 & \underline{0.2215} & 0.1942 & 0.4012 & 0.4128 & 0.4256 & 0.4815 & 0.4423 & 0.4716 \\
& RetraSyn$^{\mathrm{b}}$ & 0.1921 & 0.1645 & 0.1687 & 0.2214 & 0.2345 & 0.2381 & \underline{0.6312} & \underline{0.6245} & 0.6087 \\
& RetraSyn$^{\mathrm{p}}$ & \underline{0.2314} & 0.2156 & \underline{0.2863} & \underline{0.4211} & \underline{0.4367} & \underline{0.4523} & 0.5621 & 0.6134 & \underline{0.6415} \\
& \textsc{TrajRACE} & \textbf{0.3621} & \textbf{0.3845} & \textbf{0.4218} & \textbf{0.5318} & \textbf{0.5372} & \textbf{0.5641} & \textbf{0.6934} & \textbf{0.7028} & \textbf{0.7412} \\
\hline

\end{tabular}
}
\vspace{-10pt}
\end{table*}

\subsection{Comparison}

Table~\ref{tab:overall_utility} reports the overall synthetic-data utility on Porto, Oldenburg, and Chengdu at nominal settings \(0.5\), \(1.0\), and \(1.5\). \textsc{TrajRACE} consistently achieves the best or near-best results across most datasets, nominal settings, and metrics. Compared with the strongest baseline in each setting, 
\textsc{TrajRACE} reduces Dens-Err, Query-Err, and Diam-Err by at least 41.6\%, 25.8\%, and 22.1\%, respectively, and improves Pattern-F1 by at least 9.8\%.
These gains stem from road-network constraints and recovered mobility statistics. We make two observations.

First, for distribution metrics, including Dens-Err, Query-Err, Hot-NDCG, and K-Tau, which measure spatial density, range-query accuracy, and hotspot/ranking preservation, \textsc{TrajRACE} performs stably across the three datasets. The advantage is clear even at the smallest nominal setting \(0.5\), showing that the recovered start, segment-count, and conditional transition statistics remain useful under privacy protection. This is because \textsc{TrajRACE} keeps the current road segment as aggregation context and recovers conditional successor distributions, instead of treating movements as grid-cell visits.

Second, for structure metrics, including Trip-Err, Len-Err, Diam-Err, and Pattern-F1, which capture trip-level fidelity, segment-count length, trajectory extent, and frequent sequential patterns, the results show dataset-dependent behaviors. For example, RetraSyn$^{\mathrm{p}}$ slightly outperforms \textsc{TrajRACE} on several Trip-Err results for Chengdu, while LDPTrace performs slightly better on one Len-Err setting on Chengdu. This suggests that baselines can preserve certain single-aspect structural statistics, but their performance is less stable across datasets and structural metrics. In contrast, \textsc{TrajRACE} benefits from legal successor constraints and conditional next-hop modeling, which better preserve trajectory structures.

\subsection{Effect of Risk-Aware Protection}

\begin{table*}[t]
\centering
\caption{Effect of risk-aware protection strategies on protection and recovery quality at \(B=1.0\).}
\label{tab:riskaware_protection}
\vspace{-5pt}
\small
\setlength{\tabcolsep}{12pt}
\renewcommand{\arraystretch}{1.0}
\begin{tabular}{llccccc}
\hline
Dataset & Method
& HR-Keep $\downarrow$
& All-Keep
& Trans-JS $\downarrow$
& Bigram-JS $\downarrow$
& TSR $\uparrow$ \\
\hline

\multirow{3}{*}{Porto}
& \textsc{TrajRACE}-uniform
& 0.5542 & 0.3451 & 0.3982 & 0.5417 & 0.7025 \\
& PublicDegree-NU
& \underline{0.4386} & 0.3394 & \underline{0.3870} & \underline{0.5258} & \textbf{0.7253} \\
& \textsc{TrajRACE}
& \textbf{0.4215} & 0.3385 & \textbf{0.3754} & \textbf{0.5203} & \underline{0.7241} \\
\hline

\multirow{3}{*}{Oldenburg}
& \textsc{TrajRACE}-uniform
& 0.7831 & 0.4012 & 0.3315 & 0.2684 & 0.8762 \\
& PublicDegree-NU
& \underline{0.6687} & 0.3959 & \underline{0.3200} & \underline{0.2456} & \underline{0.8890} \\
& \textsc{TrajRACE}
& \textbf{0.6124} & 0.3948 & \textbf{0.3087} & \textbf{0.2415} & \textbf{0.8973} \\
\hline

\multirow{3}{*}{Chengdu}
& \textsc{TrajRACE}-uniform
& 0.6953 & 0.4428 & 0.3364 & 0.4612 & 0.7104 \\
& PublicDegree-NU
& \underline{0.5729} & 0.4324 & \underline{0.3230} & \underline{0.4351} & \textbf{0.7471} \\
& \textsc{TrajRACE}
& \textbf{0.5418} & 0.4315 & \textbf{0.3128} & \textbf{0.4336} & \underline{0.7358} \\
\hline

\end{tabular}
\vspace{-10pt}
\end{table*}

We study risk-aware scheduling using three controlled variants with the same SBS segmentation, reporting units, reporting-accounting rule, and recovery pipeline. \textsc{TrajRACE}-uniform uses uniform transition perturbation without risk scores; PublicDegree-NU allocates protection using only public current-segment out-degree; and \textsc{TrajRACE} further incorporates endpoint proximity, long-stay association, and low out-degree. Table~\ref{tab:riskaware_protection} reports the results.

Compared with \textsc{TrajRACE}-uniform, \textsc{TrajRACE} reduces HR-Keep by 21.8\%--23.9\%, while All-Keep changes by at most 0.012. PublicDegree-NU reduces HR-Keep, showing that public topology provides guidance for non-uniform protection. Compared with PublicDegree-NU, \textsc{TrajRACE} achieves lower Trans-JS and Bigram-JS on all datasets. Both non-uniform variants also achieve higher TSR than \textsc{TrajRACE}-uniform on all datasets. Together, these results suggest that incorporating context-level risk factors can further improve the balance between protection and mobility recovery.

\subsection{Effect of Risk Factors}

We evaluate the effect of the three risk factors used in \textsc{TrajRACE}: endpoint proximity, long-stay association, and low out-degree. The experiment is conducted on Chengdu under $B=1.0$. We compare the full model with three ablated variants: w/o End. removes endpoint-proximity risk, w/o Stay removes long-stay association risk, and w/o Deg. removes low-out-degree risk. We also include Uniform as a reference that uses the same perturbation strength for all transition units. Table~\ref{tab:risk_ablation} reports the results. 

\begin{table}[t]
\centering
\caption{Ablation study on risk components (Chengdu, \(B=1.0\)).}
\label{tab:risk_ablation}
\vspace{-5pt}
\small
\setlength{\tabcolsep}{2.5pt}
\renewcommand{\arraystretch}{1.0}
\begin{tabular}{lccccc}
\hline
Method
& \shortstack{HR-Keep \(\downarrow\)}
& \shortstack{All-Keep}
& \shortstack{Trans-JS \(\downarrow\)}
& \shortstack{Bigram-JS \(\downarrow\)}
& \shortstack{TSR \(\uparrow\)} \\
\hline
Full      & \textbf{0.5418} & 0.4315 & \textbf{0.3128} & \textbf{0.4336} & \textbf{0.7358} \\
w/o End.  & 0.6427 & 0.4382 & 0.3155 & 0.4381 & 0.7312 \\
w/o Stay  & 0.5935 & 0.4354 & 0.3142 & 0.4504 & 0.7335 \\
w/o Deg.  & 0.5712 & 0.4341 & 0.3287 & 0.4415 & 0.7086 \\
Uniform   & 0.6953 & 0.4428 & 0.3364 & 0.4612 & 0.7104 \\
\hline
\end{tabular}
\vspace{-10pt}
\end{table}

The full model achieves the best results on the primary diagnostic metrics, including HR-Keep, Trans-JS, Bigram-JS, and TSR. 
All-Keep changes only mildly across variants, suggesting that the differences mainly come from risk-factor scheduling rather than global changes in transition keep rates. 
The ablated variants further show the role of each risk factor. 
Removing endpoint proximity increases HR-Keep from 0.5418 to 0.6427, indicating that endpoint proximity helps identify and protect sensitive transitions near the start or destination of an SBS. 
Removing long-stay association gives the largest Bigram-JS among the ablated variants, suggesting that stay-related risk helps preserve second-order movement patterns. 
Removing low-out-degree risk causes the largest degradation in Trans-JS and TSR among the ablated variants: Trans-JS increases from 0.3128 to 0.3287, and TSR decreases from 0.7358 to 0.7086. 
This shows its role in preserving empirically supported successor patterns. 
The Uniform variant performs worse than the full model, indicating that using the same successor perturbation strength for all transition units is less effective than adapting protection to local transition risks.

\subsection{Effect of Mobility Statistics Recovery}
We evaluate the effect of the recovery modules used in \textsc{TrajRACE}. The experiment is conducted on Chengdu under $B=1.0$. All variants use the same privatized reports and only differ in how they recover the population-level mobility statistics $\widehat{\Phi}$. 
We compare four variants: Full, which uses the complete recovery pipeline; w/o Bucket Corr., which removes bucket correction for noisy bucket labels; w/o Shrinkage, which removes local-global shrinkage; and Single-channel, which uses one fixed perturbation channel instead of the mixed channel. 
Table~\ref{tab:recovery_ablation} reports the results. The full model achieves the best results on Trans-JS, Bigram-JS, and TSR. This shows that the complete recovery pipeline improves conditional transition recovery, preserves second-order movement patterns, and better follows observed transition support. Since transition reports are affected by both successor perturbation and noisy bucket reporting, directly recovering transition statistics without correcting these effects leads to larger recovery bias.

\begin{table}[t]
\centering
\caption{Recovery-module ablation on Chengdu at \(B=1.0\).}
\label{tab:recovery_ablation}
\vspace{-5pt}
\small
\setlength{\tabcolsep}{2pt}
\renewcommand{\arraystretch}{1}
\begin{tabular}{lccc}
\hline
Variant
& \shortstack{Trans-JS \(\downarrow\)}
& \shortstack{Bigram-JS \(\downarrow\)}
& \shortstack{TSR \(\uparrow\)} \\
\hline
Full              & \textbf{0.3128} & \textbf{0.4336} & \textbf{0.7358} \\
w/o Bucket Corr.  & 0.3275 & 0.4402 & 0.7281 \\
w/o Shrinkage     & 0.3186 & 0.4358 & 0.7295 \\
Single-channel    & 0.3224 & 0.4451 & 0.7233 \\
\hline
\end{tabular}
\vspace{-10pt}
\end{table}

The ablated variants show the role of each recovery module. 
Removing bucket correction increases Trans-JS from 0.3128 to 0.3275, showing the need to correct bucket-label distributions before conditional successor recovery. 
Removing shrinkage degrades all three metrics, indicating that local-global shrinkage stabilizes sparse current-segment estimates. 
Single-channel performs worst on Bigram-JS and TSR because a fixed successor-perturbation channel cannot model heterogeneous execution buckets under the current segment.

\subsection{Hyperparameter Sensitivity}

\noindent\textbf{Effect of the shrinkage parameter.}
The shrinkage parameter $\tau_\pi$ controls how much the context-specific bucket estimate is smoothed toward the global bucket estimate during recovery. Fig.~\ref{fig:tau_sensitivity} reports the results under $\tau_\pi\in\{5,20,50\}$. \textsc{TrajRACE} changes only slightly across different $\tau_\pi$ values on all three metrics, showing that it is not highly sensitive to this hyperparameter. A smaller $\tau_\pi$ sometimes gives lower Trans-JS because it keeps more local transition information. However, $\tau_\pi=20$ gives more stable Bigram-JS and TSR, suggesting that moderate smoothing helps preserve second-order transition structures and empirical transition-support consistency. When $\tau_\pi$ increases to 50, several metrics degrade slightly. This is because excessive smoothing weakens local transition patterns. 

\begin{figure}[t]
\centering
\includegraphics[width=\columnwidth]{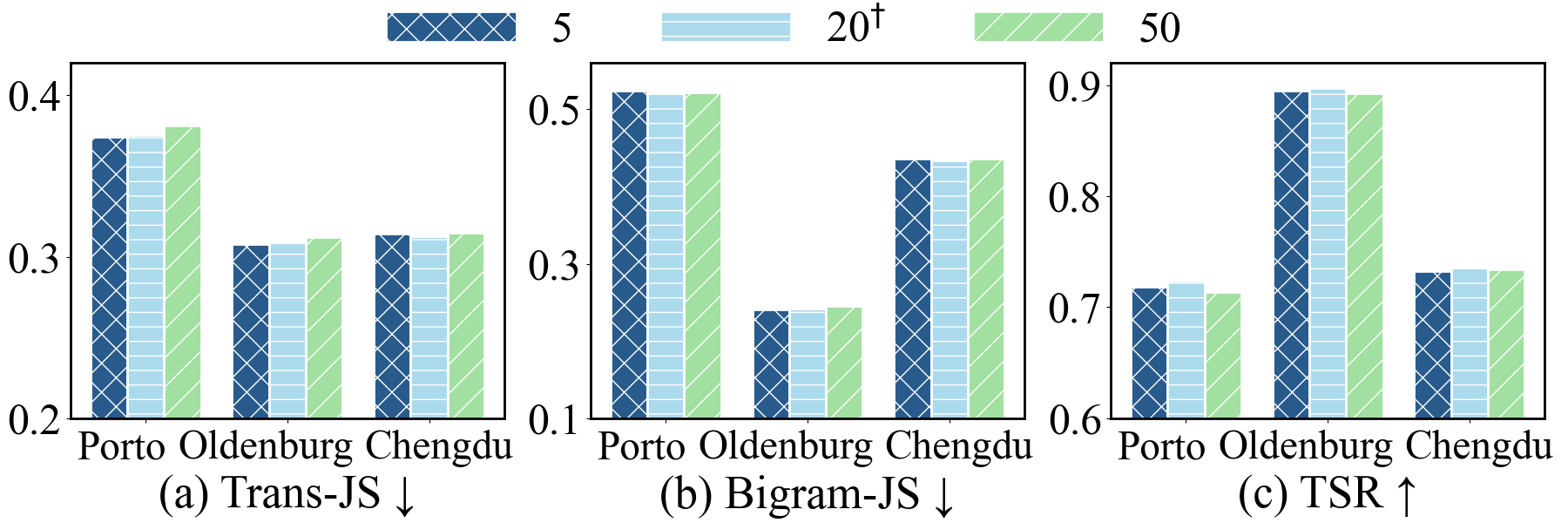}
\vspace{-20pt}
\caption{Effect of the shrinkage parameter \(\tau_\pi\).}
\vspace{-10pt}
\label{fig:tau_sensitivity}
\end{figure}

\noindent\textbf{Effect of risk weights.}
The risk weights control the importance of endpoint proximity, long-stay association, and low out-degree in the transition risk score. Fig.~\ref{fig:risk_weight_sensitivity} reports the results under four pre-normalized weight settings: Bal. $(1\!:\!1\!:\!1)$, End. $(2\!:\!1\!:\!1)$, Stay $(1\!:\!2\!:\!1)$, and Deg. $(1\!:\!1\!:\!2)$. HR-Keep, Trans-JS, and Bigram-JS vary only slightly across these settings, indicating that \textsc{TrajRACE} is not highly sensitive to the manual choice of risk weights. Degree-heavy usually reduces HR-Keep because low-out-degree segments expose fewer legal next-hop choices and need stronger protection. Endpoint-heavy and Stay-heavy improve Trans-JS and Bigram-JS on some datasets because endpoint and stay information helps identify sensitive local movements. The balanced setting does not always give the best value for every metric, but it provides stable overall performance across datasets.

\begin{figure}[t]
\centering
\includegraphics[width=\columnwidth]{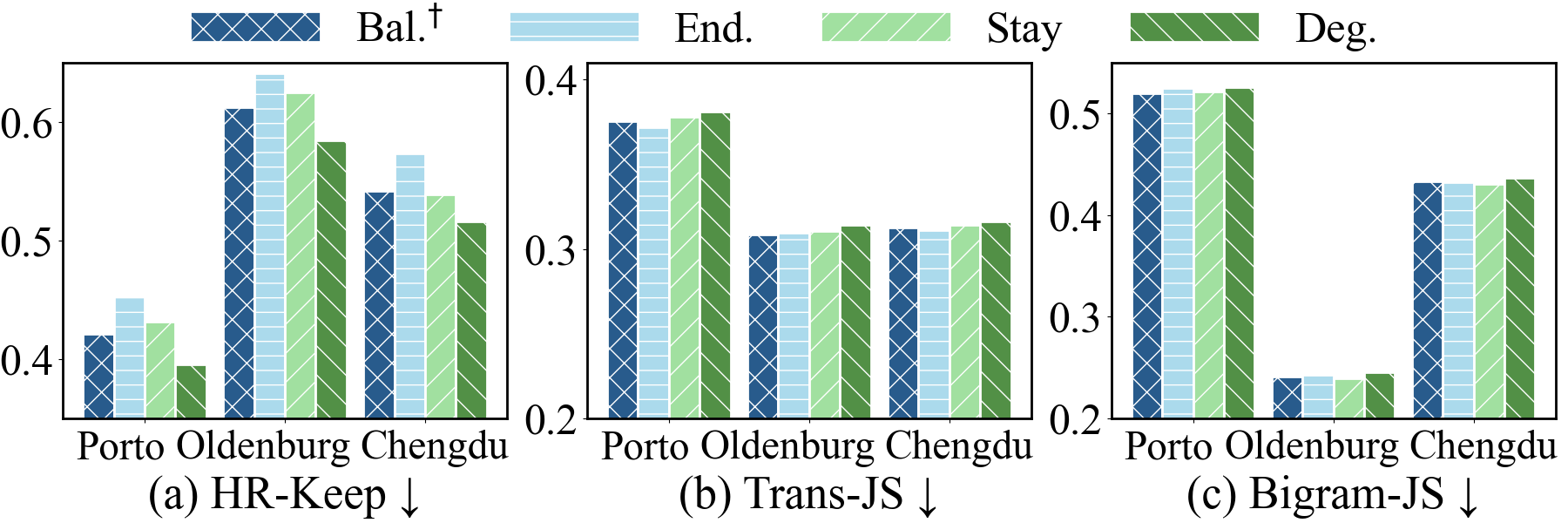}
\vspace{-20pt}
\caption{Effect of relative risk weights.}
\vspace{-10pt}
\label{fig:risk_weight_sensitivity}
\end{figure}

\noindent\textbf{Effect of bucket number.}
The bucket number $K$ controls the granularity of transition protection levels. Fig.~\ref{fig:k_sensitivity} reports the results under $K\in\{2,3,4\}$. The performance changes only moderately across different $K$, suggesting that \textsc{TrajRACE} is stable under different protection granularities. When $K=2$, Trans-JS is often lower, but HR-Keep is the highest on all three datasets. This is because coarse buckets make recovery more stable but weaken the separation between high-risk and low-risk transitions. When $K=4$, High-risk Keep is consistently the lowest, but Trans-JS and Bigram-JS are slightly worse in most cases. This is because fine-grained buckets separate risks better but make the recovery problem sparser. The default setting $K=3$ provides a balanced trade-off between high-risk protection and recovery quality.

\begin{figure}[t]
\centering
\includegraphics[width=\columnwidth]{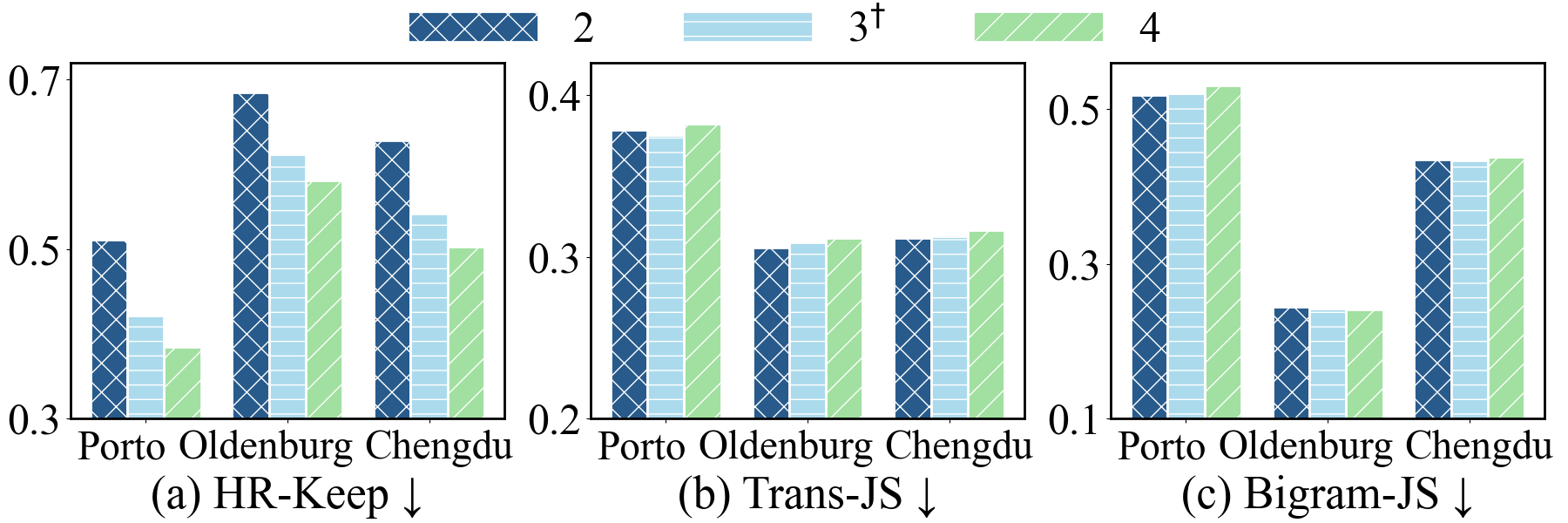}
\vspace{-20pt}
\caption{Effect of the number of privacy buckets \(K\).}
\vspace{-10pt}
\label{fig:k_sensitivity}
\end{figure}

\subsection{Efficiency Analysis}
Fig.~\ref{fig:efficiency} reports the runtime of all methods on Porto, Oldenburg, and Chengdu at the nominal setting \(1.0\). For the baselines, we report the end-to-end runtime of their implemented pipelines, while for \textsc{TrajRACE}, we further report client-side risk-aware reporting, server-side recovery, and synthetic trajectory generation. Trace, Syn$^{\mathrm b}$, and Syn$^{\mathrm p}$ denote LDPTrace, RetraSyn$^{\mathrm b}$, and RetraSyn$^{\mathrm p}$, respectively.

\textsc{TrajRACE} takes 379.22s, 402.95s, and 466.07s on Porto, Oldenburg, and Chengdu, respectively. Its runtime mainly comes from client-side risk-aware reporting, which computes transition risks, assigns privacy buckets, and generates privatized reports. In contrast, server-side recovery remains lightweight, taking only 54.24s, 26.18s, and 41.27s, respectively. Although \textsc{TrajRACE} is not the fastest method, its total runtime remains within several minutes on all datasets and is practical for offline privacy-preserving trajectory synthesis.

\begin{figure}[t]
\centering
\includegraphics[width=1.0\columnwidth]{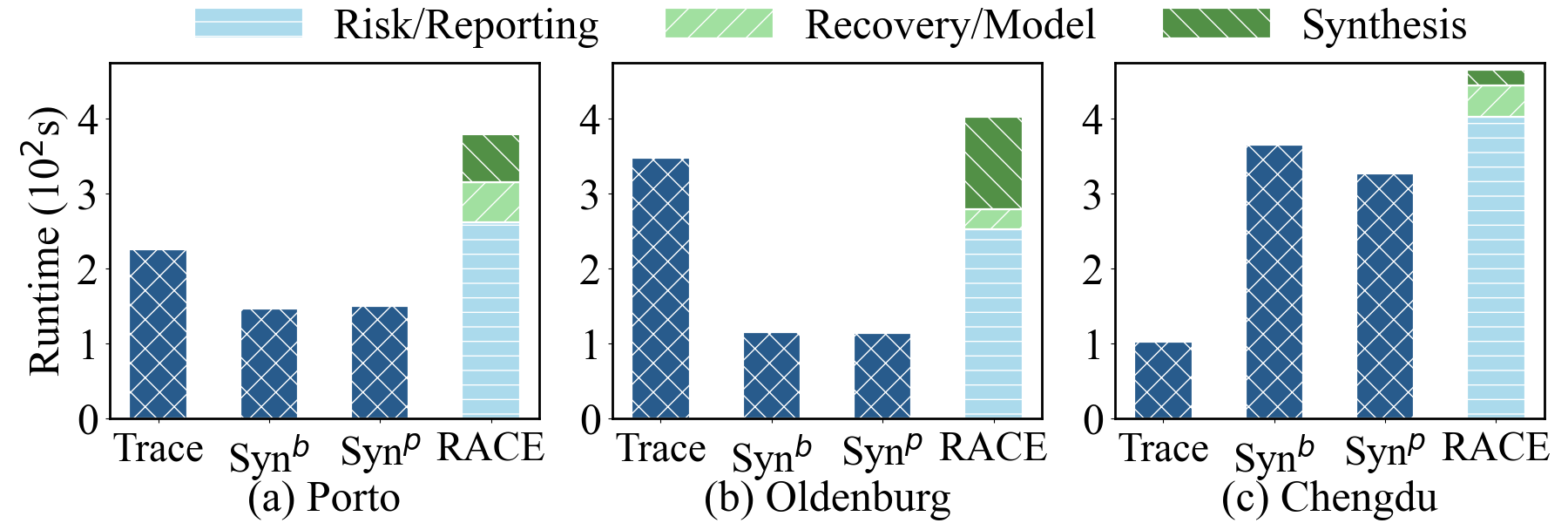}
\vspace{-20pt}
\caption{Efficiency comparison across datasets at the nominal setting \(1.0\).}
\vspace{-10pt}
\label{fig:efficiency}
\end{figure}

\section{Conclusion}\label{sec:conclusion}

We proposed \textsc{TrajRACE}, a risk-aware event-level local privacy framework for road-network trajectory synthesis under an untrusted server. It protects true successors under public current-segment contexts and combines risk-aware protection with bucket-corrected recovery. Experiments on Porto, Oldenburg, and Chengdu show that \textsc{TrajRACE} improves synthetic-data utility over representative baselines and reduces HR-Keep by 21.8\%--23.9\% relative to the uniform variant.

Several limitations remain. The privacy guarantee is conditional on a disclosed current segment and applies to individual transition reports, while $B$ is only an SBS reporting-accounting cap; multi-SBS composition remains future work. Recovery uses a context-level bucket-mixture approximation, which may affect accuracy but not the conditional privacy guarantee. Future work will study multi-SBS composition, current-segment protection, more exact recovery, online synthesis, and sequence-aware risk models.

\section*{Acknowledgment}
This work was supported in part by the National Natural Science Foundation of China (Grants No. 62272379, T2341003, and 62441227), the Natural Science Basic Research Program of Shaanxi Province (Grant No. 2025JC-JCQN-081), the XJTU Research Fund for AI Science (Grant No. 2025YXYC004), and the Youth Innovation Team of Shaanxi Universities.

\section*{AI-Generated Content Acknowledgment}
During the preparation of this work, the authors used ChatGPT solely for language editing and grammar enhancement. The authors reviewed the output and take full responsibility for the final manuscript.

\bibliographystyle{IEEEtran}
\bibliography{refs_icde}
\balance

\end{document}